\documentclass{amsart}

\usepackage{faktor}
\usepackage{tabularx}
\usepackage{float}

\usepackage{boldline}

\usepackage{hyperref}
\usepackage{amsmath}
\usepackage{amssymb,amsfonts}
\usepackage[foot]{amsaddr}

\usepackage{tikz}
\usetikzlibrary{trees}
\usetikzlibrary{matrix, fit, backgrounds}
\usetikzlibrary{shapes.misc, positioning}
\usepackage{tikz-qtree}
\usepackage{forest}

\usepackage{stmaryrd}

\usepackage{xstring}
\usepackage{enumerate}

\renewcommand{\phi}{\varphi}
\newcommand{\todo}[1]{}

\newcommand{\lhs}{\mathit{LHS}}
\newcommand{\rhs}{\mathit{RHS}}

\newcommand{\set}[1]{\{#1\}}
\newcommand{\quot}[1]{\emph{#1}}
\newcommand{\e}[1]{\emph{#1}}
\renewcommand{\emptyset}{\varnothing}

\newtheorem*{theorem*}{Theorem}
\newtheorem{theorem}{Theorem}
\newtheorem{thm}[theorem]{Theorem}
\newtheorem{proposition}[theorem]{Proposition}

\newtheorem{observation}[theorem]{Observation}
\newtheorem{remark}[theorem]{Remark}
\newtheorem{corollary}[theorem]{Corollary}

\newtheorem{lemma}[theorem]{Lemma}
\newtheorem{lem}[theorem]{Lemma}
\newtheorem{fact}[theorem]{Fact}
\newtheorem{claim}[theorem]{Claim}

\newtheorem{definition}[theorem]{Definition}
\newtheorem{defi}[theorem]{Definition}
\newtheorem{example}[theorem]{Example}
\newtheorem{ex}[theorem]{Example}

\newcommand{\pow}{\mathcal{P}}

\newcommand{\negquot}{\mathsf{Neg}}
\newcommand{\posquot}{\mathsf{Pos}}

\newcommand{\msks}{\mathsf{str}}

\newcommand{\eqac}{=_{\mathit{AC}}}

\newcommand{\moddec}{\mathit{T}}
\newcommand{\tree}{\moddec}

\newcommand{\Var}{\mathcal{V}}

\newcommand{\SAT}{\mathsf{SAT}}

\newcommand{\CNF}{\mathsf{CNF}}

\newcommand{\con}{\wedge}
\newcommand{\dis}{\vee}
\newcommand{\imp}{\supset}
\newcommand{\entails}[1]{\underset{#1}{\Rightarrow}}

\newcommand{\consys}{\annarr \con}
\newcommand{\dissys}{\annarr \dis}

\newcommand{\wk}{\mathsf{w}}
\newcommand{\cntr}{\mathsf{c}}
\newcommand{\switch}{\mathsf{s}}
\newcommand{\medial}{\mathsf{m}}

\newcommand{\Copy}{\mathsf{d}_{\wedge}}
\newcommand{\Dopy}{\mathsf{d}_{\vee}}

\newcommand{\annarr}[1]{\underset{#1}{\rightarrow}}

\newcommand{\PP}{\mathbf{P}}
\newcommand{\NP}{\mathbf{NP}}
\newcommand{\coNP}{\mathbf{coNP}}
\newcommand{\cohier}[2]{\Pi^{#1}_{#2}}
\newcommand{\ndhier}[2]{\Sigma^{#1}_{#2}}

\newcommand{\Ver}{V}
\newcommand{\Edg}{E}
\newcommand{\undirset}[1]{\binom{#1}{2}}
\newcommand{\pfour}{P_4}
\newcommand{\restr}{\hspace{-.3em}\upharpoonright\hspace{-.3em}}

\newcommand{\web}{\mathcal{W}}

\newcommand{\compl}[1]{\overline{#1}}

\newcommand{\mc}{\mathit{MC}}
\newcommand{\ms}{\mathit{MS}}

\newcommand{\vbcengraph}[1]{\raisebox{-0.4\height}{#1}}

\tikzstyle{vertex}=[circle,minimum size=10pt,inner sep=0pt,text height=1.5ex,text depth=.25ex]
\tikzstyle{module}=[ellipse,minimum height=9pt,minimum width=15pt,inner sep=1.5pt,text height=1.5ex,text depth=.25ex, fill=gray!30]
\tikzstyle{g} = [draw,thick,dotted,-,green]
\tikzstyle{lg} = [draw,thick,dotted,-,green]
\tikzstyle{r} = [draw,thick,-,red]
\tikzstyle{b} = [draw,thick,-,black]
\tikzstyle{w} = [draw,thick,-,white]
\tikzstyle{u} = [draw,thin,-,black,opacity=0]

\newcommand{\Mod}[1]{
	{\begin{tikzpicture}[baseline=(v1.base)]
	\node[vertex] (v1) at (0,0) {$ {#1} $};
	\node[draw, fit=(v1)]{};
	\end{tikzpicture}}
}

\newcommand{\TwoGraph}[2]{ 
	{\begin{tikzpicture}[baseline=(v1.base)]
		\node[vertex] (v1) at (0, 0) {$\StrBetween[1,2]{,#1,}{,}{,}$};
		\node[vertex]  (v2) at (1, 0) {$\StrBetween[2,3]{,#1,}{,}{,}$};
		\draw[#2]  (v1) -- (v2) ;
		\end{tikzpicture} }}

	\newcommand{\TwoGraphLabel}[3][]{ 
		{\begin{tikzpicture}[baseline=(v1.base)]
			\node[vertex] (v1) at (0, 0) {$\StrBetween[1,2]{,#2,}{,}{,}$};
			\node[vertex]  (v2) at (1, 0) {$\StrBetween[2,3]{,#2,}{,}{,}$};
			\draw[#3]  (v1) -- (v2) node [midway,above=-1pt] {$\scriptstyle #1$};
			\end{tikzpicture} }}

	\newcommand{\TwoGraphVertMod}[2]{{\begin{tikzpicture}[baseline=(v1.base)]
			\node[vertex] (v1) at (0, 0) {$\StrBetween[1,2]{,#1,}{,}{,}$};
			\node[vertex]  (v2) at (0, 1) {$\StrBetween[2,3]{,#1,}{,}{,}$};
			\draw[#2]  (v1) -- (v2) ;
						\node[draw, fit=(v1) (v2)] {};
			\end{tikzpicture} }}

\newcommand{\redge}[3][]{\TwoGraphLabel[\smash{#1}]{#2,#3}r}
\newcommand{\gedge}[3][]{\TwoGraphLabel[\smash{#1}]{#2,#3}g}

\newcommand{\FourGraphPath}[7]{ 
	\raisebox{-0.4\height}{\begin{tikzpicture}
		\node[vertex]  (v1) at (0, 0) {$\StrBetween[1,2]{,#1,}{,}{,}$};
		\node[vertex]  (v2) at (1, 0) {$\StrBetween[2,3]{,#1,}{,}{,}$};
		\node[vertex] (v3) at (2, 0) {$\StrBetween[3,4]{,#1,}{,}{,}$};
		\node[vertex]  (v4) at (3, 0) {$\StrBetween[4,5]{,#1,}{,}{,}$};
		\draw[#2]  (v1) -- (v2) ;
		\draw[#3,bend left]  (v1) edge (v3);
		\draw[#4, bend left] (v1) edge (v4);
		\draw[#5]  (v2) -- (v3);
		\draw[#6,bend left] (v2) edge (v4);
		\draw[#7]  (v3) -- (v4);
		\end{tikzpicture}} }

\newcommand{\ThreeGraph}[4]{ 
	{\begin{tikzpicture}[baseline=(v1.base)]
		\node[vertex] (v1) at (0, 0.6667) {$\StrBetween[1,2]{,#1,}{,}{,}$};
		\node[vertex]  (v2) at (1, 1.3333) {$\StrBetween[2,3]{,#1,}{,}{,}$};
		\node[vertex] (v3) at (1 ,0 ) {$\StrBetween[3,4]{,#1,}{,}{,}$};
		\draw[#2]  (v1) -- (v2) ;
		\draw[#3] (v1) -- (v3) ;
		\draw[#4] (v2) -- (v3) ;
		\end{tikzpicture} }}

\newcommand{\FiveGraphCircCont}[2]{
	\raisebox{-0.45\height}{\begin{tikzpicture}
		\node[vertex] (v1) at (0,1) {$\StrBetween[1,2]{,\fir,}{,}{,}$};
		\node[vertex] (v2) at (0.951,0.309) {$\StrBetween[2,3]{,\fir,}{,}{,}$};
		\node[vertex] (v3) at (0.588,-0.809) {$\StrBetween[3,4]{,\fir,}{,}{,}$};
		\node[vertex] (v4) at (-0.588,-0.809) {$\StrBetween[4,5]{,\fir,}{,}{,}$};
		\node[vertex] (v5) at (-0.951,0.309) {$\StrBetween[5,6]{,\fir,}{,}{,}$};
		\draw[\sec] (v1) -- (v2) ;
		\draw[\thi] (v1) -- (v3);
		\draw[\fou] (v1) -- (v4);
		\draw[\fif] (v1) -- (v5);
		\draw[\six] (v2) -- (v3);
		\draw[\sev] (v2) -- (v4);
		\draw[\eig] (v2) -- (v5) ;
		\draw[\nin] (v3) -- (v4) ;
		\draw[#1] (v3) -- (v5) ;
		\draw[#2] (v4) -- (v5) ;
		\end{tikzpicture} }}
\newcommand{\FiveGraphCirc}[9]{
	\def\fir{#1}
	\def\sec{#2}
	\def\thi{#3}
	\def\fou{#4}
	\def\fif{#5}
	\def\six{#6}
	\def\sev{#7}
	\def\eig{#8}
	\def\nin{#9}
	\FiveGraphCircCont
}

\newcommand{\FourGraph}[7]{ 
	\raisebox{-0.4\height}{\begin{tikzpicture}
		\node[vertex]  (v1) at (0, 1.3333) {$\StrBetween[1,2]{,#1,}{,}{,}$};
		\node[vertex]  (v2) at (1.3333, 1.3333) {$\StrBetween[2,3]{,#1,}{,}{,}$};
		\node[vertex] (v3) at (0, 0) {$\StrBetween[3,4]{,#1,}{,}{,}$};
		\node[vertex]  (v4) at (1.3333, 0) {$\StrBetween[4,5]{,#1,}{,}{,}$};
		\draw[#2]  (v1) -- (v2) ;
		\draw[#3]  (v1) -- (v3);
		\draw[#4] (v1) -- (v4);
		\draw[#5]  (v2) -- (v3);
		\draw[#6] (v2) -- (v4);
		\draw[#7]  (v3) -- (v4);
		\end{tikzpicture}} }

\newcommand{\SixGraphCont}[7]{
	\raisebox{-0.3\height}{\begin{tikzpicture}
		\tikzstyle{vertex}=[circle,minimum size=10pt,inner sep=0pt]
		\tikzstyle{g} = [draw,thick,dotted,-,green]
		\tikzstyle{r} = [draw,thick,-,red]
		\tikzstyle{u} = [draw,thick,-,white]
		\node[vertex] (v1) at (0, 1) {$\StrBetween[1,2]{,\fir,}{,}{,}$};
		\node[vertex]  (v2) at (1, 2) {$\StrBetween[2,3]{,\fir,}{,}{,}$};
		\node[vertex] (v3) at (3,2) {$\StrBetween[3,4]{,\fir,}{,}{,}$};
		\node[vertex] (v4) at (1,0) {$\StrBetween[4,5]{,\fir,}{,}{,}$};
		\node[vertex] (v5) at (3,0) {$\StrBetween[5,6]{,\fir,}{,}{,}$};
		\node[vertex] (v6) at (4,1) {$\StrBetween[6,7]{,\fir,}{,}{,}$};
		\draw[\sec]  (v1) -- (v2) ;
		\draw[\thi] (v1) -- (v3) ;
		\draw[\fou] (v1) -- (v4) ;
		\draw[\fif] (v1) -- (v5) ;
		\draw[\six] (v1) -- (v6) ;
		\draw[\sev] (v2) -- (v3) ;
		\draw[\eig] (v2) -- (v4) ;
		\draw[\nin] (v2) -- (v5) ;
		\draw[#1] (v2) -- (v6) ;
		\draw[#2] (v3) -- (v4) ;
		\draw[#3] (v3) -- (v5) ;
		\draw[#4] (v3) -- (v6) ;
		\draw[#5] (v4) -- (v5) ;
		\draw[#6] (v4) -- (v6) ;
		\draw[#7] (v5) -- (v6) ;
		\end{tikzpicture} }}
\newcommand\SixGraph[9]{
	\def\fir{#1}
	\def\sec{#2}
	\def\thi{#3}
	\def\fou{#4}
	\def\fif{#5}
	\def\six{#6}
	\def\sev{#7}
	\def\eig{#8}
	\def\nin{#9}
	\SixGraphCont
}

\newcommand{\SevenModContCont}[4]{
	\raisebox{-0.3\height}{\begin{tikzpicture}[scale = 1.6	]
		\tikzstyle{vertex}=[circle,minimum size=10pt,inner sep=0pt]
		\tikzstyle{g} = [draw,thick,dotted,-,green]
		\tikzstyle{r} = [draw,thick,-,red]
		\tikzstyle{u} = [draw,thick,-,white]
		\node[vertex] (v1) at (0,0) {$\StrBetween[1,2]{,\fir,}{,}{,}$};
		\node[vertex] (v2) at (1,0) {$\StrBetween[2,3]{,\fir,}{,}{,}$};
		\node[vertex] (v3) at (2,0) {$\StrBetween[3,4]{,\fir,}{,}{,}$};
		\node[vertex] (v4) at (3,0) {$\StrBetween[4,5]{,\fir,}{,}{,}$};
		\node[vertex] (v5) at (2.5,1) {$\StrBetween[5,6]{,\fir,}{,}{,}$};
		\node[vertex] (v6) at (0,1) {$\StrBetween[6,7]{,\fir,}{,}{,}$};
		\node[vertex] (v7) at (1,1) {$\StrBetween[7,8]{,\fir,}{,}{,}$};
		\draw[\sec]  (v1) edge (v2) ;
		\draw[\thi, bend right] (v1) edge (v3) ;
		\draw[\fou, bend right] (v1) edge (v4) ;
		\draw[\fif] (v1) -- (v5) ;
		\draw[\six] (v1) -- (v6) ;
		\draw[\sev] (v2) -- (v3) ;
		\draw[\eig, bend right] (v2) edge (v4) ;
		\draw[\nin] (v2) -- (v5) ;
		\draw[\ten] (v2) edge (v6) ;
		\draw[\ele] (v3) -- (v4) ;
		\draw[\twel] (v3) -- (v5) ;
		\draw[\thirt] (v3) edge (v6) ;
		\draw[\fourt] (v4) -- (v5) ;
		\draw[\fift] (v4) -- (v6) ;
		\draw[\sixt, bend right ] (v5) edge (v6) ;
		\draw[\sevent] (v7) -- (v1) ;
		\draw[\eight] (v7) edge (v2);
		\draw[#1] (v7) -- (v3);
		\draw[#2] (v7) -- (v4);
		\draw[#3] (v7) -- (v5);
		\draw[#4] (v7) -- (v6);
		\end{tikzpicture} }
}
\newcommand{\SevenModCont}[9]{
	\def\ten{#1}
	\def\ele{#2}
	\def\twel{#3}
	\def\thirt{#4}
	\def\fourt{#5}
	\def\fift{#6}
	\def\sixt{#7}
	\def\sevent{#8}
	\def\eight{#9}
	\SevenModContCont
}
\newcommand{\SevenMod}[9]{
	\def\fir{#1}
	\def\sec{#2}
	\def\thi{#3}
	\def\fou{#4}
	\def\fif{#5}
	\def\six{#6}
	\def\sev{#7}
	\def\eig{#8}
	\def\nin{#9}
	\SevenModCont
}

\newcommand{\SevenGraphModuleContCont}[4]{
	\raisebox{-0.3\height}{\begin{tikzpicture}
		\tikzstyle{vertex}=[circle,minimum size=10pt,inner sep=0pt]
		\tikzstyle{g} = [draw,thick,dotted,-,green]
		\tikzstyle{r} = [draw,thick,-,red]
		\tikzstyle{u} = [draw,thick,-,white]
		\node[vertex] (v1) at (0, -1) {$\StrBetween[1,2]{,\fir,}{,}{,}$};
		\node[vertex] (v2) at (1, 2) {$\StrBetween[2,3]{,\fir,}{,}{,}$};
		\node[module] (v3) at (3,2) {$\StrBetween[3,4]{,\fir,}{,}{,}$};
		\node[vertex] (v4) at (1,0) {$\StrBetween[4,5]{,\fir,}{,}{,}$};
		\node[vertex] (v5) at (3,0) {$\StrBetween[5,6]{,\fir,}{,}{,}$};
		\node[module] (v6) at (4,-1) {$\StrBetween[6,7]{,\fir,}{,}{,}$};
		\node[vertex] (v7) at (2,-2) {$\StrBetween[7,8]{,\fir,}{,}{,}$};
		\draw[\sec,bend left]  (v1) edge (v2) ;
		\draw[\thi,bend left] (v1) edge (v3) ;
		\draw[\fou] (v1) -- (v4) ;
		\draw[\fif] (v1) -- (v5) ;
		\draw[\six] (v1) -- (v6) ;
		\draw[\sev] (v2) -- (v3) ;
		\draw[\eig] (v2) -- (v4) ;
		\draw[\nin] (v2) -- (v5) ;
		\draw[\ten,bend left] (v2) edge (v6) ;
		\draw[\ele] (v3) -- (v4) ;
		\draw[\twel] (v3) -- (v5) ;
		\draw[\thirt,bend left] (v3) edge (v6) ;
		\draw[\fourt] (v4) -- (v5) ;
		\draw[\fift] (v4) -- (v6) ;
		\draw[\sixt] (v5) -- (v6) ;
		\draw[\sevent] (v7) -- (v1) ;
		\draw[\eight,bend left] (v7) edge (v2);
		\draw[#1, bend right] (v7) edge (v3);
		\draw[#2] (v7) -- (v4);
		\draw[#3] (v7) -- (v5);
		\draw[#4] (v7) -- (v6);
		\end{tikzpicture} }
}
\newcommand{\SevenGraphModuleCont}[9]{
	\def\ten{#1}
	\def\ele{#2}
	\def\twel{#3}
	\def\thirt{#4}
	\def\fourt{#5}
	\def\fift{#6}
	\def\sixt{#7}
	\def\sevent{#8}
	\def\eight{#9}
	\SevenGraphModuleContCont
}
\newcommand{\SevenGraphModule}[9]{
	\def\fir{#1}
	\def\sec{#2}
	\def\thi{#3}
	\def\fou{#4}
	\def\fif{#5}
	\def\six{#6}
	\def\sev{#7}
	\def\eig{#8}
	\def\nin{#9}
	\SevenGraphModuleCont
}

\title[Beyond formulas-as-cographs: Boolean Graph Logic]{Beyond formulas-as-cographs: an extension of Boolean logic to arbitrary graphs}
	\author{Cameron Calk$^1$}
	\address{$^1$Laboratoire d'Informatique de l'\'Ecole Polytechnique, Palaiseau, 91120, France.}
		\author{Anupam Das$^2$}
		\address{$^2$University of Birmingham, Birmingham, B15 2TT, UK.}
		\author{Tim Waring$^3$}
		\address{$^3$University of Copenhagen, Copenhagen, 2100, Denmark.}
	
	\date{\today}

\begin{document}
						
			\begin{abstract}
				We propose a graph-based extension of Boolean logic called Boolean Graph Logic (BGL). Construing formula trees as the cotrees of cographs, we may state semantic notions such as evaluation and entailment in purely graph-theoretic terms, whence we recover the definition  of BGL. Naturally, it is conservative over usual Boolean logic. 
				
				Our contributions are the following:
				\begin{enumerate}
					\item We give a natural semantics of BGL based on Boolean relations, i.e.\ it is a multivalued semantics, and show the adequacy of this semantics for the corresponding notions of entailment.
					\item We show that the complexity of evaluation is $\NP$-complete for arbitrary graphs (as opposed to $\mathbf{ALOGTIME}$-complete for formulas), while entailment is $\Pi^p_2$-complete (as opposed to $\coNP$-complete for formulas).
					\item We give a `recursive' algorithm for evaluation by induction the modular decomposition of graphs. (Though this is not polynomial-time, cf.~point (2) above).
					\item We characterise evaluation in a game-theoretic setting, in terms of both static and sequential strategies, extending the classical notion of positional game forms beyond cographs.
					\item We give an axiomatisation of BGL, inspired by deep-inference proof theory, and show soundness and completeness for the corresponding notions of entailment.
				\end{enumerate}
	
				One particular feature of the graph-theoretic setting is that it escapes certain no-go theorems such as a recent result of Das and Strassburger \cite{DasStra:nolinearsys15,DasStra:16:On-Linea:uk}, that there is no linear axiomatisation of the linear fragment of Boolean logic (equivalently the multiplicative fragment of Japaridze's Computability Logic or Blass Game Semantics for Multiplicative Linear Logic).
			\end{abstract}
		\maketitle
	
		\section{Introduction}
		\emph{Boolean} logic, a.k.a classical propostional logic, lies at the heart of multiple areas, including algebra, proof theory and computational complexity.
		Axiomatisations and {proof systems} for this logic generally manipulate Boolean formulas, built from $\neg, \vee,\wedge$ or some other adequate basis of connectives.
		These include the classical Hilbert-Frege style axiomatic systems and Gentzen style sequent and natural deduction systems, as well as more recent methodologies such as Belnap's display logic, cf.~\cite{Belnap82a}, and Guglielmi's deep inference, cf.~\cite{Gugl:14:Deep-Inf:fj}.
		The latter, along with Girard's linear logic \cite{Girard87}, has also been responsible for a more graph-based viewpoint on proof theory, for instance via proof nets \cite{Girard87}, expansion proofs \cite{Miller84}, atomic flows \cite{GuglGund:07:Normalis:lr} and combinatorial proofs \cite{Hughes:PWS}.
		Nonetheless such systems still fundamentally deal with the same underlying structure, formulas, with various annotations/decorations.

One viewpoint of formulas is to construe formula trees as the cotrees of cographs, leading to the notion of `co-occurrence graph' (in Boolean function theory, e.g.\ \cite{CraHam:BFbook11}) or `relation web' (in structural proof theory, e.g.\ \cite{Gugl:06:A-System:kl}).
This viewpoint has been exploited several times in order to prove fundamental properties of Boolean logic. A notable result is seminal theorem of Gurvich:
\begin{theorem*}
	[\cite{gurvich1977repetition}]
	A Boolean function $f$ is computed by a read-once formula if and only if every minterm and maxterm of $f$ has singleton intersection.
\end{theorem*}

More recently, one the current authors, Das, proved with Strassburger that the linear fragment of Boolean logic admits no polynomial-time axiomatisation (unless $\coNP = \NP$), crucially exploiting the graph-theoretic viewpoint \cite{DasStra:nolinearsys15,DasStra:16:On-Linea:uk}.
This result is important since it immediately implies the non-axiomatisability of multiplicative fragments of Japaridze's Computability logic, cf.~\cite{Japaridze05:intro-to-cirquent-calc}, and  Blass' game semantics for linear logic \cite{Blass92}, resolving a problem that was open since the '90s (see \cite{Japaridze17}).

One feature of the graph-theoretic setting is that semantic notions such as evaluation and entailment may be framed in purely graph-theoretic terms, e.g.\ as exploited in \cite{gurvich1977repetition,DasStra:16:On-Linea:uk}. In particular, these characterisations are meaningful for all graphs, not just the cographs that correspond to formulas.
In this work we further develop this idea into a bona fide logical system, which we call `Boolean Graph Logic' (BGL).
We establish various fundamental properties of BGL, from the viewpoints of complexity theory, game theory and proof theory, using graph theoretic tools throughout.
BGL is conservative over usual Boolean logic (as expected), and semantically it corresponds to a natural extension to Boolean relations rather than Boolean functions.
One particular feature of BGL is that entailment is much more fine-grained, admitting interpolations that are impossible in usual Boolean logic. An example of this is given in Section~\ref{subsect:case-study}, where a minimal 10 variable linear inference from \cite{Das:13:Rewritin:uq} is decomposable in BGL.
\todo{is it really minimal?}

At least one motivation of this work is the issue of (linearly) axiomatising the linear fragment of Boolean logic.
The impossibility result of \cite{DasStra:nolinearsys15,DasStra:16:On-Linea:uk} crucially exploited the fact that the graphs corresponding to formulas are $\pfour$-free (i.e.\ cographs). In particular, their critical Lemma~5.8 does not scale to arbitrary graphs. It is natural to ask whether this impossibility result can be extended to BGl or whether BGL might admit a linear axiomatisation.
One way of looking at this motivation is with the following question: can we establish a Boolean proof theory without any structural behaviour, such as duplication and erasure? 
A positive answer would be complementary to the established orthodoxy in structural proof theory.

Our contributions are the following:
\begin{enumerate}
	\item We give a natural semantics of BGL based on Boolean relations, i.e.\ it is a multivalued semantics, and show the adequacy of this semantics for the corresponding notions of entailment.
	\item We show that the complexity of evaluation is $\NP$-complete for arbitrary graphs (as opposed to $\mathbf{ALOGTIME}$-complete for formulas), while entailment is $\Pi^p_2$-complete (as opposed to $\coNP$-complete for formulas).
	\item We give a `recursive' algorithm for evaluation by induction the modular decomposition of graphs. (Though this is not polynomial-time, cf.~point (2) above).
	\item We characterise evaluation in a game-theoretic setting, in terms of both static and sequential strategies, extending the classical notion of positional game forms beyond cographs.
	\item We give an axiomatisation of BGL, inspired by deep-inference proof theory, and show soundness and completeness for the corresponding notions of entailment.
\end{enumerate}
%
%
%
%
%
%
%

		\subsection{History and related work}
		The ideas behind the Boolean Graph Logic project were already mentioned in the aforementioned paper \cite{DasStra:16:On-Linea:uk}, where Section~9 proposed the basic notions of entailment as the possible basis of a proof theory on arbitrary graphs.
		The research behind BGL properly began in 2016 when the author Calk conducted his Bachelor's thesis project under the supervision of the author Das \cite{calk}. This comprised mainly of the results from Sections~\ref{sect:bgl} and \ref{sect:moddecomp}.
		The results of Section~\ref{sect:complexity} were established by Das and have been circulated in private correspondence \cite{das-note}.
		The author Waring conducted his Master's thesis under the supervision of Das in 2019 \cite{waring}, comprising of the results of Sections~\ref{sect:games} and parts of \ref{sect:bgl} and \ref{sect:proof-system}, as well as supplementary material found in the appendix.
		
		In parallel other groups of authors have become interested in the prospect of logic and proof theory based on arbitrary graphs.
		In particular, Acclavio, Horne and Strassburger have investigated a structural proof theory of arbitrary graphs, from a linear logic perspective \cite{AccHorStr:20}. They obtain a cut-elimination result for a conservative extension of multiplicative linear logic, via the `splitting' method of deep inference proof theory.  Their syntax of graphs is identical to ours, only interpreting conjunction and disjunction by their multiplicative variants. However their logic and BGL seem to be incomparable, at the level of validity, distinguishing these two graph settings from their restriction to formulas. 
			It would be interesting to establish a graph level semantics of their logic to more fully understand the differences between these two approaches; relevant work in this direction includes recent work of Seiller and Nguyen, who developed a model of multiplicative linear logic with a form of nondeterminism via `interaction graphs' \cite{NguSei18}.

%
%
%

			Beyond structural proof theory, there are several systems that operate with objects other than formulas.
		These include circuit-based systems (e.g.\ \cite{Japaridze05:intro-to-cirquent-calc} and \cite{jerabek:dual-wphp}), algebraic systems (e.g.\ \cite{BeameIKPP94} and \cite{CuttingPlaneProofs-CookCT87}) and systems operating with forms of decision diagrams (e.g.\ \cite{AtseriasKV04} and \cite{Buss0K20}).
		Of these, algebraic systems are particularly interesting (and thematically relevant to this work) since they extend the usual Boolean semantics conservatively to an arithmetic setting.
		
		\subsection{Perspectives on Boolean Graph Logic}
		Given that this article collects various research carried out by the three authors, it is natural that different readers will take interest in different parts. 
		We give the mutual dependencies between sections in Figure~\ref{fig:dependencies}, and we give some possible readings of this article below:
		\begin{itemize}
			\item Complexity theoretic viewpoint: Sections~\ref{sect:prelims}, \ref{sect:bgl} and \ref{sect:complexity}.
			\item Game theoretic viewpoint: Sections~\ref{sect:prelims}, \ref{sect:bgl}, \ref{sect:moddecomp} and \ref{sect:games}.
			\item Proof theoretic viewpoint: Sections~\ref{sect:prelims}, \ref{sect:bgl}, \ref{sect:moddecomp} and \ref{sect:proof-system}.
		\end{itemize}
		It is our intention that this article be an introduction to BGL and serve as a reference for later research.
		
		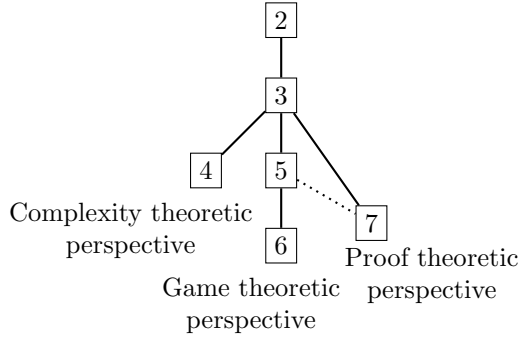
\begin{figure}[t]
			\[
			\begin{tikzpicture}
			\node[draw] (prelims) at (0,0) {\ref{sect:prelims}};
			\node[draw] (bgl) at (0,-1) {\ref{sect:bgl}};
			\node[vertex,align=center] (compview) at (-2,-3) {Complexity theoretic \\ perspective};
			\node[draw] (complexity) at (-1,-2) {\ref{sect:complexity}};
			\node[draw] (moddecomp) at (0,-2) {\ref{sect:moddecomp}};
			\node[vertex,align=center] (gameview) at (-0.4,-4) {Game theoretic \\ perspective};
			\node[draw] (games) at (0,-3) {\ref{sect:games}};
			\node[vertex,align=center] (proofview) at (2,-3.6) {Proof theoretic \\ perspective};
			\node[draw] (system) at (1.2,-2.7) {\ref{sect:proof-system}};
			\draw[b] (prelims) -- (bgl);
			\draw[b] (bgl) -- (complexity);
			\draw[b] (bgl) -- (moddecomp);
			\draw[b] (bgl) -- (system);
			\draw[b] (moddecomp) -- (games);
			\draw[b,dotted] (moddecomp) -- (system);
			\end{tikzpicture}
			\]
		\caption{Dependencies of sections in this paper, and the viewpoints of Boolean Graph Logic given by different readings of the paper. For Section~\ref{sect:proof-system}, familiarity with the content of Section~\ref{sect:moddecomp} is helpful but not strictly necessary. }
		\label{fig:dependencies}
		\end{figure}

		\section{Preliminaries}
		\label{sect:prelims}
		Throughout this work let us fix a set $\Var$ of \emph{variables} that will be used in various contexts, e.g.\ nodes of a graph, Boolean variables etc. 
		This overloading is intentional as we will often identify these objects later on.
		We write $\undirset X$ for the set of unordered pairs of a set $X$.
		
		\subsection{Graphs}
		In this work we deal with undirected finite graphs that are simple and loopless. 
		We specify a graph $G$ by a pair $(V,E)$ such that $V$ is a finite set of \emph{nodes} or \emph{vertices}, typically a subset of $\Var$, and $E \subseteq \undirset{V}$ the set of \emph{edges}.
		The \emph{dual} or \emph{complement} 
		$\compl{G}$ of a graph $G = (V,E)$ is $\left(V, \undirset{V} \setminus E \right)$.
		The \emph{size} of a graph $G$, written $|G|$, is just its number of vertices, i.e.\ $|V|$.
		
		We will use standard graph-theoretic terminology, e.g.\ neighbours, adjacency, cliques, stable sets, connected etc.
		Sometimes we will write $\Ver(G)$ (or $\Edg (G)$) to denote the set of vertices (respectively edges) of a graph $G$ when it is not completely specified.
		
		For convenience, we will often (but not always) draw graphs with red solid lines to indicate edges and green dotted lines to indicate non-edges, to allow for incomplete information to be presented. 
		I.e.\ we may say ``\emph{$\TwoGraph{x,y}{r}$ in $G$}'' to denote that $\{x,y\}\in \Edg(G)$ and ``\emph{$\TwoGraph{x,y}{g}$ in $G$}'' to denote that $\{x,y\} \notin \Edg(G)$.

		\begin{example}
			[Paths, cycles and complete graphs]
			Let $V_n = \{x_1 , \dots , x_n \}$.
			We will use the following graphs throughout this work.
			\begin{itemize}
				\item $ P_n := \left( V_n , \{ \{x_1, x_2\}, \dots , \{x_{n-1}, x_n\} \}\right) $ is the \emph{$n$-path}. For instance, here are some (equivalent) representations of the $4$-path, $\pfour$:
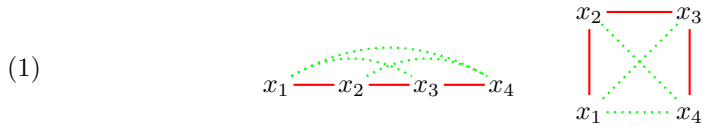
				\begin{equation}
				\label{eqn:paths}
				\FourGraphPath{x_1,x_2,x_3,x_4}rggrgr
				\qquad
				\FourGraph{x_2,x_3,x_1,x_4}rrggrg
				\end{equation}
				\item $ C_n := P_n \cup \{x_n, x_1 \} $ is the \emph{$n$-cycle}. For instance, here are the $4$-cycle $C_4$ and $5$-cycle $C_5$, respectively:
				\begin{equation}
				\label{eqn:cycles}
				\FourGraph{x_2,x_3,x_1,x_4}rrggrr
				\qquad
				\FiveGraphCirc{x_1,x_2,x_3,x_4,x_5}rggrrggrgr
				\end{equation} 
				\item \( K_n := \left( V_n , \undirset{V_n} \right) \) is the \emph{$n$-complete-graph}. For instance here are the graphs $K_3$ and $K_5$, respectively:
				\[
				\ThreeGraph{x_1,x_2,x_3,x_4}rrr
				\qquad
				\FiveGraphCirc{x_1,x_2,x_3,x_4,x_5}rrrrrrrrrr
				\]
			\end{itemize}
		\end{example}

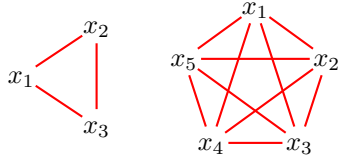
		\begin{definition}
			[Subgraphs]
			A graph $G$ is a \emph{subgraph} of a graph $H$, written $G \leq H $, if $\Ver (G) \subseteq \Ver (H)$ and $\Edg (G) \subseteq \Edg (H)$.
			For $X \subseteq \Ver(G)$ we say that $ G\restr X := \left( X, \undirset{X} \cap \Edg(G) \right)$ is the subgraph of $G$ \emph{induced} by $X$.
			We will sometimes simply write $X$ instead of $G\restr X$, as an abuse of notation, when the ambient graph $G$ is clear from context.
		\end{definition}

		\begin{definition}
			[Homomorphisms and freeness]
			For graphs $G$ and $H$, we say that $\phi : \Ver(G) \to \Ver (H)$ is a \emph{homomorphism} if $\{x,y\} \in \Edg(G)$ implies that $\{\phi(x), \phi(y) \} \in \Edg(H)$, or $\phi(x) = \phi(y)$.\footnote{This caveat is because we deal with loopless graphs. However note that this convention does not affect the notions of clique, stable set and connectedness that we rely on throughout this work.}
			If $\phi$ is a bijection then we say it is an \emph{isomorphism}.
			
			We say that $G$ is \emph{$H$-free} if no induced subgraph of $G$ is isomorphic to $H$.
		\end{definition}

	\begin{remark}
		[Variants of homomorphisms]
		\label{rmk:rg-homomorphisms}
		It would have been quite natural for us to consider homomorphisms that further satisfy the following dual property: if $\{\phi(x), \phi(y) \} \notin V(H)$ then $\{x,y\}\notin V(G)$, i.e.\ $\phi$ \emph{reflects} green edges.
		This is pertinent because of the notion of disjunctive entailment $\entails \dis$ we introduce in the next section, but we do not develop this notion of homomorphism formally.
		
		We could rather consider homomorphisms that preserve more structure, for instance `preserving green edges' too. Such a property is satisfied by the `quotient homomorphisms' we discuss in Section~\ref{sect:moddecomp}.
		More detailed investigations studying the links between logical phenomena and classes of homomorphisms have been carried out in \cite{RalphS19}.
	\end{remark}
		
		We will often speak about graphs up to isomorphism, when there is no ambiguity.
		
		\begin{example}
			[Cographs]
			A \emph{cograph} is a graph for which every induced subgraph is either disconnected or its complement is disconnected.
			
			For instance, the $P_4$ (cf.~\eqref{eqn:paths}) is not a cograph, but the $C_4$ is (cf.~\eqref{eqn:cycles}). Notice that the $C_5$ is not a cograph since it has a $P_4$ as an induced subgraph.
			This means that the $C_5$ is not $\pfour$-free.
		\end{example}
		
		We have the following well-known characterisation of cographs, immediately yielding an efficient algorithm for cograph-recognition.
		
		\begin{fact}
			\label{prop:cograph-iff-p4free}
			$G$ is a cograph if and only if it is $\pfour$-free.
		\end{fact}
	\todo{cite?}
		
		Finally, we introduce some topological aspects of graphs that we use throughout.
		\begin{definition}
			[Cliques, stable sets and connectedness]
			For a graph $G$, a subset of its nodes $X \subseteq \Ver (G)$ is a \emph{clique} if $G \restr X$ is isomorphic to some $K_n$.
			If $G\restr X$ is isomorphic to some $\compl{K}_n$ then we say that $X$ is a \emph{stable set}.
			We say that a clique (stable set) is \emph{maximal} in $G$ if it cannot be extended to a larger clique (respectively stable set).
			We write $\mc (G)$ ($\ms (G)$) for the set of maximal cliques (respectively stable sets) of $G$.
			
			A graph $G$ is \emph{connected} if, for every $x,y \in \Ver(G)$, there is an induced subgraph isomorphic to some $P_n$ containing $x$ and $y$. Otherwise $G$ is \emph{disconnected}.
			$G$ is \emph{co-connected} if, for every $x,y \in G$, there is an induced subgraph isomorphic to some $\compl{P}_n$ containing $x$ and $y$. Otherwise $G$ is \emph{co-disconnected}.
		\end{definition}

		Note, in particular, that a graph may be both connected and coconnected, in particular the $P_4$:
		\[
		\FourGraph{ x_2,x_3 ,x_1 ,x_4 }rrggrg
		\]
		Notice that co-connectedness can be seen visually as `green' connectedness, whereas connectedness is `red' connectedness.

		The following result is immediate from the definitions of homomorphisms and cliques:
		\begin{observation}
			[Homomorphisms preserve cliques]
			\label{prop:homo-pres-refl-cli-stab}
			Let $\phi : \Ver (G) \to \Ver (H)$ be a homomorphism from a graph $G$ to a graph $H$.
			Then, if $S$ is a clique in $G$ then $\phi(S)$ is a clique in $H$. 
		\end{observation}
%
%

Following on from Remark~\ref{rmk:rg-homomorphisms}, for homomorphisms that reflect green edges we have a dual result to that above: if $T $ is stable in $H$ then $\phi^{-1}(T)$ is stable in $G$.
On the other hand if green edges are preserved then so are stable sets.

Notice, however, that maximality of a clique is not preserved by a homomorphism. A very simple example is given by the map below:
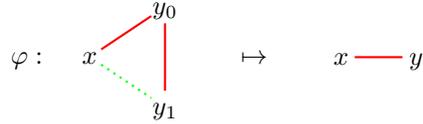
\[
\phi: \quad \ThreeGraph{x,y_0,y_1}rgr
\qquad \mapsto \qquad
\TwoGraph{x,y}{r}
\]
by the homomorphism $\{x \mapsto x, y_0 \mapsto y, y_1 \mapsto y \}$. Here the maximal clique $\{y_0,y_1 \}$ on the LHS is mapped to the clique $y$ on the RHS, but this is clearly not maximal.
This example also shows that, a priori, stable sets are not reflected by homomorphisms: $\{y\}$ is stable on the RHS but its preimage $\{y_0,y_1 \}$ is clearly not stable on the LHS.

		\subsection{Boolean logic}
		We will consider \emph{positive Boolean formulas} (henceforth simply \emph{formulas}), built as follows:
		\begin{itemize}
			\item Any variable $x\in \Var$ is a formula;
			\item If $A$ and $B$ are formulas then so is $(A\dis B)$;
			\item If $A$ and $B$ are formulas then so is $(A\con B)$.
		\end{itemize}
		
		We call a formula \emph{linear} or \emph{read-once} if each variable occurs at most once in it.\footnote{The nomenclature `linear' comes from term rewriting, e.g.\ \cite{Terese}, where as `read-once' comes from complexity theory and Boolean function theory, e.g.\ \cite{crama2011boolean}.}
		Formulas compute Boolean functions in the usual way. We identify Boolean assignments $\Var \to \{0,1\}$, with subsets of $\Var$ as expected: $\alpha:\Var \to \{0,1\}$ is identified with the subset $\{x \in \Var : \alpha(x) = 1 \} \in \pow (\Var)$.
		In this way, a Boolean function is a map $\pow (\Var) \to \{0,1\}$.
		Note that, throughout this work, we assume that the support of assignments is finite, so that only finite subsets of variables need be considered.
		
		 Formally, for a formula $A$, we write $A(X)$ for its Boolean output on input $X$, defined recursively as follows:
		\[
		\begin{array}{rcl}
		x(X) & := & \begin{cases}
		1 & x \in X \\
		0 & \text{otherwise}
		\end{cases} \\
		(A\dis B)(X) & :=& \max (A(X),B(X)) \\
				(A\con B)(X) & :=& \min (A(X),B(X)) 
		\end{array}
		\]
		
		In our setting, where we do not admit negation, the functions computed by formulas are \emph{monotone}, i.e.\ if $X \subseteq Y $ then $ A(X) \leq A(Y)$.
		For us this is sufficiently general since any Boolean tautology $A(x_1, \dots, x_n, \compl x_1, \dots \compl x_n)$, with all literals displayed, is equivalent to a monotone implication $\bigwedge\limits^n_{i=1} (x_i \dis {x}_i') \Rightarrow A(\vec x, \vec x')$. Note that if $A$ is read-once then the resulting monotone implication is a linear inference, i.e.\ both the LHS and RHS are read-once on the same variables.

		More generally, we may speak of arbitrary Boolean functions $f: \pow (\Var) \to \{0,1\}$ being monotone if $X \subseteq Y \implies f(X) \leq f(Y)$.
		Again, since we are only considering finite Boolean functions, such functions $f$ are determined by their actions on some finite set of variables $V_n = \{x_1 , \dots, x_n \}$ that they `depend on'.

		We introduce now a standard semantic abstraction of monotone Boolean functions, which we will later link to graph-theoretic notions.
		\begin{definition}
			[Minterms and maxterms]
			Let $f : \pow(\Var) \to \{0,1\}$ be a monotone Boolean function.
			\begin{itemize}
				\item A \emph{minterm} of $f$ is a minimal set $X\subseteq \Var$ such that $f(X) = 1$.
				\item Dually, a \emph{maxterm} of $f$ is a minimal set $X$ such that $f(\compl X) = 0$.
			\end{itemize}
			Since we identify formulas with the monotone Boolean functions they compute, we may also speak of the minterms and maxterms of a formula.
		\end{definition}
		
		Notice that minterms correspond to the terms of the irredundant disjunctive normal form of a function, while maxterms correspond to the clauses of the irredundant conjunctive normal form of a function (see, e.g., \cite{crama2011boolean}).
		Consequently we have the following (well-known) characterisations of evaluation and entailment in terms of their minterms and maxterms:
		
		\begin{proposition}
		[Characterisations of evaluation]
		\label{prop:evaluation-characterisation}
		For a monotone Boolean function $f:\pow (\Var) \to \{0,1\}$, we have the following:
			\begin{enumerate}
				\item $f(X) = 1$ if and only if there is a minterm $S$ of $f$ such that $S \subseteq X$.
				\item $f(X) = 0$ if and only if there is a maxterm $T$ of $f$ such that $T\cap X = \emptyset$.
			\end{enumerate}
		\end{proposition}
		
		\begin{proposition}
		[Characterisations of entailment]
		\label{prop:entailment-characterisation}
			For a monotone Boolean function $f:\pow (\Var) \to \{0,1\}$, the following are equivalent
			\begin{enumerate}
				\item $f \leq g$, i.e.\ if $f(X) = 1$ then $g(X) = 1$.
				\item For each minterm $S$ of $f$ there is a minterm $S'$ of $g$ such that $S' \subseteq S$.
				\item For each maxterm $T$ of $g$ there is a maxterm $T'$ of $f$ such that $T' \subseteq T$.
			\end{enumerate}
		\end{proposition}
	
	Both results are folklore, but proofs may be found in, e.g., \cite{DasStra:16:On-Linea:uk}, along with several examples.	
	Naturally, the book \cite{crama2011boolean} also constitutes good reference material.
	
		
		\subsection{Relation webs}
		Now we introduce a useful abstraction of Boolean formulas, that relates them (and their semantics) to the graph theoretic notions we introduced earlier. 

		\begin{definition}
			[Least common connectives and relation webs]
			For a linear formula $A$ containing (distinct) variables $x$ and $y$, the \emph{least common connective} (lcc) of $x$ and $y$ in $A$ is the main connective, $\dis$ or $\con$, of the smallest subformula of $A$ containing both $x$ and $y$.
			
			The \emph{relation web}, or simply \emph{web}, of $A$, written $\web(A)$, is the graph defined as follows:
			\begin{itemize}
				\item The nodes of $\web(A)$ are the variables of $A$.
				\item The edges of $\web (A)$ are those $\{x,y\}$ where the lcc of $x$ and $y $ is $\con$ in $A $.
			\end{itemize}
		\end{definition}
	
	Relation webs are also known as `co-occurrence' graphs, though we follow the nomenclature from structural proof theory here.
	
	\begin{example}
		\label{ex:formula-web-5}
The formula $((v \vee (w\wedge x)  ) \vee y)\wedge z $ has the following relation web:
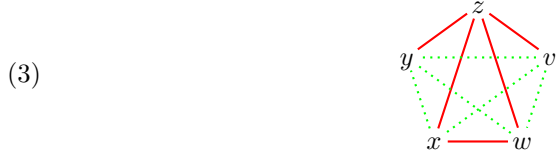
\begin{equation}
\label{eqn:fivegraphexampleweb}
\FiveGraphCirc{z,v,w,x,y}rrrrgggrgg
\end{equation}
Notice that we can also see the web as drawing a `red edge' between variables whose lcc is $\con$ and a `green edge' between variables whose lcc is $\dis$.

The graph above is also the relation web of, e.g., $z \con (((x\wedge w) \dis y) \dis v )$.
	\end{example}
		
		In fact, webs represent precisely the quotients of linear formulas modulo associativity and commutativity of $\dis$ and $\con$:
		\begin{fact}
			\label{fact:web-ac}
			$\web(A) = \web(B)$ if and only if $A $ and $B$ are equivalent modulo the following equational theory:\footnote{As usual, equational theories operate under `deep inference', i.e.\ if $A \eqac A' $ then $B[A] \eqac B[A']$ for any formula context $B[\cdot]$.}
			\[
\begin{array}{rcl}
			A \dis B &\eqac & B \dis A \\
			A \con B &\eqac & B \con A 
\end{array}
\qquad
\begin{array}{rcl}
A\dis(B\dis C) & \eqac & (A \dis B) \dis C \\
A\con(B\con C) & \eqac & (A \con B) \con C 
\end{array}
			\]
		\end{fact}
This fact is well-known and may be proved by induction on the length of a $\eqac$-derivation.
		The graph classification of relation webs is well-known:
		\begin{fact}
			Any relation web is a cograph. Conversely, any cograph (with nodes in $ \Var$) is the web of some linear formula.
		\end{fact}
	This is readily proved by induction on the size of graphs, using either the underlying formula structure of a web or the fact that cograph-ness is closed under taking induced subgraphs.
	
	Notice that the above fact also gives us another characterisation of relation webs in terms of forbidden induced subgraphs:
	\begin{corollary}
		[of Proposition~\ref{prop:cograph-iff-p4free}]
		Any relation web is a $\pfour$-free. Conversely, any $\pfour$-free graph (with nodes in $ \Var$) is the web of some linear formula.
	\end{corollary}
This alternative characterisation makes it abundantly clear that the relation webs really only span a small subset of all possible graphs: almost all graphs contain $\pfour$s, in a standard sense, by a simple counting argument:

\begin{remark}
	[Counting cographs]
	Appealing to the probabilistic method, note that the uniform distribution on graphs (of fixed size) is induced by independently assigning edges to pairs of variables with probability $\frac 1 2 $. 
	In this case, for any fixed four nodes, the chance that they do not form a $\pfour$ is some fixed $\varepsilon <1$. Thus for a graph with, say, $4n$ nodes, the chance there is no $\pfour$ is bounded above by $\varepsilon^n$ by partitioning the nodes into sets of $4$ (since $\pfour$-ness of disjoint sets are independent events).
	Therefore the class of cographs, as we increase the number of nodes, is sparse in the class of all graphs. 
\end{remark}

		As we have mentioned the point of this work is to study an extension of the notion of web to arbitrary graphs.
		In order to do so we will need characterisations of logical concepts in terms of graphs.
		Namely, we have the following results from \cite{DasStra:nolinearsys15,DasStra:16:On-Linea:uk}:		
		\begin{proposition}
		[Characterisation of maximal cliques and stables sets of webs]
		\label{prop:cliques-minterms-stable-maxterms}
		We have the following, for any linear formula $A$:
		\begin{enumerate}
		\item $S$ is a minterm of $A$ if and only if it is a maximal clique of $\web(A)$. 
		\item $T$ is a maxterm of $A$ if and only if it is a maximal stable set of $\web(A)$.
		\end{enumerate}
		\end{proposition}
		
		This result is proved by a routine induction on the structure of the formula $A$.
	 This gives us purely graph theoretic characterisations of evaluation and entailment, thanks to Propositions~\ref{prop:evaluation-characterisation} and \ref{prop:entailment-characterisation}, which will induce the graph logic we define in the next section.
	 Notice also that this means that distinct webs correspond to distinct Boolean functions, since two webs are the same just if they have the same maximal cliques and stable sets, just if they have the same minterms and maxterms, just if they compute the same Boolean function.
	 Since we also have that webs quotient formulas exactly by associativity and commutativity of $\dis$ and $\con$, cf.~Fact~\ref{fact:web-ac}, we also have the following well-known result:
	 
	 \begin{corollary}
	 	[E.g., \cite{DasStra:16:On-Linea:uk}]
	 	Read-once/linear formulas compute the same Boolean function if and only if they are equivalent modulo $\eqac$.
	 \end{corollary}

		\begin{ex}
			Revisiting Example~\ref{ex:formula-web-5},
			notice that the two formulas whose web is \eqref{eqn:fivegraphexampleweb} are equivalent modulo $\eqac$, so let us write $f$ for the Boolean function they compute. 
			Notice that $X = \{w,y,z \}$ contains the maximal clique $\{y,z\}$, and so $f(X)=1$.
			Dually, $Y = \{v,w,y \}$ is disjoint from the stable set $\{z \}$, and so $f(Y)=0$.
			Both of these facts are readily verifiable by computing the output of $f$ using the formulas from Example~\ref{ex:formula-web-5}.
\end{ex}
	
\begin{example}
	[Switch and medial]
	Here are two common rules from deep inference proof theory \cite{BrunTiu:01:A-Local-:mz}:
	\[
\begin{array}{rrcl}
	\switch : & x \con (y \dis z) & \to & (x\con y)  \dis z \\
	\medial: & (w\con x) \dis (y \con z) & \to & (w \dis y ) \con (x \dis z)
\end{array}
	\]
	In terms of relation webs, the two rules above induce the following action on graphs:
	\[
	\switch : \ 
	{\ThreeGraph{x, y,z}rrg} \ \ \to \ \ {\ThreeGraph{x,y,z}rgg }
	\qquad \qquad
	\medial : \ 
	\FourGraph{w,x,y,z}rggggr \ \ \to \ \ \FourGraph{w,x,y,z}rgrrgr 
	\]
	We may verify that these rules are indeed sound by checking that every maximal clique on the LHS has a subset that is a maximal clique on the RHS, under Propositions~\ref{prop:cliques-minterms-stable-maxterms} and \ref{prop:entailment-characterisation}. 
	For $\switch$ the two maximal cliques $\{x,y \}$ and $\{x,z\}$ are sent to $\{x,y\}$ and $\{z\}$ respectively.
	For $\medial$ the two maximal cliques $\{w,x\}$ and $\{y,z\}$ are sent to themselves.
	Notice that $\switch$ removes edges and reduces the size of cliques, whereas $\medial $ adds edges but does not increase the size of any maximal cliques, it only adds new ones, here $\{w,z\}$ and $\{x,y\}$.
	
	We may also verify soundness using the dual condition, that every maximal stable set of the RHS has a subset maximally stable in the LHS.
	The argument is similar and we leave the details to the reader.
\end{example}

Several other examples on viewing evaluation and entailment in the setting of relation webs can be found in \cite{DasStra:16:On-Linea:uk}.

		\section{Boolean graph logic}
		\label{sect:bgl}
		We will now consider arbitrary graphs, not just the cographs, and study the notions of evaluation and entailment induced on them by the characterisations of the previous section.
		We call the resulting framework \emph{Boolean Graph Logic} (BGL).

		We present the basic logic in the next subsection, give some properties of evaluation and entailment in BGL in Subsection~\ref{subsect:det-and-tot}.
		We present a case study of these concepts in action in Subsection~\ref{subsect:case-study}, and we conclude in Subsection~\ref{subsect:det-tot-pfour-free} by characterising the graphs which compute deterministic and total Boolean relations, i.e.\ the Boolean functions.
		
			We will only present the case of `linear' graphs here, where each variable is associated to at most one node.
		This is sufficiently general for all of our theoretical development, but we nonetheless give an extension to the nonlinear case in Section~\ref{sect:proof-system}.
		
		\subsection{Evaluation and entailment}
	\label{subsect:evaluation-entailment}
		
		While formulas compute Boolean functions, our extension to graphs has a natural semantics based on Boolean \emph{relations}.
		
		\begin{definition}
			[Evaluation]
			We construe graphs (with vertices in $\Var$) as binary relations $\pow(\Var) \times \{0,1\}$, defined as follows:
			\begin{itemize}
								\item $G(X,0)$ if $\exists T \in \ms(G) . X \cap T = \emptyset$.
				\item $G(X,1)$ if $\exists S \in \mc(G) . X \supseteq S$.
			\end{itemize}
		\end{definition}
		Notice that, under the results of the previous section, we have that, for any linear formula $A$ and $X\subseteq \Var$:
		\begin{itemize}
			\item $\web(A)(X,0)$ iff $A(X) = 0$; and,
			\item $\web(A)(X,1)$ iff $A(X) = 1$.
		\end{itemize}
		Thus the notion of evaluation above is indeed conservative over usual evaluation on formulas.
		
		One immediate observation is that while evaluation is total and deterministic for cographs, since they are the webs of linear formulas which compute Boolean functions, this is no longer necessarily the case for arbitrary graphs.
				In general they compute Boolean relations which might be nondeterministic, partial, or both.
		Let us consider an example of evaluation being nonfunctional, known already from \cite{DasStra:16:On-Linea:uk}.
		
		\begin{example}
			[Evaluating $\pfour$]
			Let us again consider the following graph $G$ that is isomorphic to the $\pfour$:
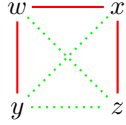
			\[
			\FourGraph{w,x,y,z}rrggrg
			\]
			We have that $G(\{w,y \},1)$, since $\{w,y\}\in \mc(G)$, and $G(\{x,y \},0)$, since $\{x,y\}$ is disjoint from $\{w,z\}\in \ms(G)$. It is not difficult to see that $G$ is indeed functional on these assignments, i.e.\ it is not the case that $G(\{w,y\},0)$ or $G(\{x,y\},1)$.
			
			On the other hand, we have that both $G(\{w,x\},0)$ and $G(\{w,x\},1)$, since $\{w,x\} \in \mc(G)$ and $\{w,x\}$ is disjoint from $\{y,z\}\in \ms(G)$.
			So $G$ is not deterministic on this assignment.
			Furthermore, we have that neither $G(\{y,z\},0)$ nor $G(\{y,z\},1)$, by inspection of the maximal cliques and maximal stable sets of $G$. So $G$ is not total on this assignment.
		\end{example}

		
		Entailment in BGL is similarly induced by our previous characterisations:
		
		\begin{definition}
			[Entailment]
			We define binary relations $\entails \con$ and $\entails \dis$ on graphs (with vertices in $\Var$) as follows:
			\begin{itemize}
				\item $G \entails \con H$ if $\forall S \in \mc (G) . \exists S'\in \mc(H) . S'\subseteq S$.
				\item $G \entails \dis H$ if $\forall T \in \ms (H). \exists T'\in \ms(G) . T'\subseteq T$. 
			\end{itemize}
		\end{definition}
		
		Again, the richer setting of arbitrary graphs means that previously equivalent notions of entailment no longer coincide, which is why we distinguish $\entails \dis$ and $\entails \con$.
		Let us consider the following example, which partly appeared in \cite{DasStra:16:On-Linea:uk}, to highlight the difference between the two forms of entailment.
		
		\begin{example}
			[$5$-cycle and $5$-path]
			We have the following relationships between  $P_5$ and $C_5$:
			\[
			\FiveGraphCirc{v,w,x,y,z}rgggrggrgr
			\quad \begin{array}{c} \not \entails \dis \\
			\entails \con \\
			\underset{\vee}{\Leftarrow} \\
			\not\underset{\wedge}{\Leftarrow} \end{array} \quad 
			\FiveGraphCirc{v,w,x,y,z}rggrrggrgr
			\]
			The argugments are as follows:
			\begin{itemize}
				\item $\not \entails \dis$ since $\{v,x\}$ and $\{z,x\}$ are maximally stable on the RHS, but on the LHS only the maximal stable set $\{v,x,z \}$ concerns these three nodes.
				\item $\entails \con$ since every maximal clique of the LHS is just an edge which is preserved in the RHS.
				\item $	\underset{\vee}{\Leftarrow} $ by sending $\{v,x,z\} $ to $\{v,x \}$ or $\{z,x\}$, and every other maximal stable set of the LHS to itself in the RHS.
				\item $\not\underset{\wedge}{\Leftarrow} $ since the maximal clique $\{v,z\}$ on the RHS has no subset that is a maximal clique on the LHS.
			\end{itemize}
		\end{example}
		
		Despite the fact that we now have two notions of evaluation and two notions of entailment, they are still `compatible' in a natural sense.
		Before proving this, we make the following remark to simplify proofs throughout this work.

\begin{remark}
	[Duality]
	\label{rmk:duality}
	Many of the arguments in this work follow by `duality'. By this we mean not only that an argument is similar to a previous one, but further that there is a formal reduction. 
	Such reduction is exhibited by the following facts:
	\begin{itemize}
		\item $G(X,1) $ if and only if $\compl G (\compl X , 0)$.
		\item $G \entails \con H$ if and only if $\compl H \entails \dis \compl G$.
	\end{itemize}
The proofs are routine and left to the reader.
\end{remark}
		
		\begin{theorem}
			[Adequacy]
			\label{thm:sound-complete}
			For graphs $G$ and $H$, we have:
			\begin{equation}
\label{eqn:sound-complete-con}
			G \entails{\con} H \quad \text{iff} \quad \forall X . ( \text{ if } G(X,1) \text{ then } H(X,1) )
			\end{equation}
			\begin{equation}
			\label{eqn:sound-complete-dis}
			G \entails{\dis} H \quad \text{iff} \quad \forall X . (\text{ if }H(X,0) \text{ then }G(X,0))
			\end{equation}
		\end{theorem}
	\begin{proof}
We prove only \eqref{eqn:sound-complete-con}, the argument for \eqref{eqn:sound-complete-dis} being dual.

For the left-right implication, suppose $G \entails \con H$ and $G(X,1)$. Then there is some $S \in \mc (G)$ that is contained in $X$. Since $G \entails \con H$, there is some $S' \in \mc (H)$ s.t.\ $S' \subseteq S$ and so also $S' \subseteq X$, so we indeed have that $H(X,1)$.

For the right-left implication, suppose the RHS and let $S \in \mc (G)$. Clearly we have that $G(S,1)$, since $S\subseteq S$, and so by assumption we have that $H(S,1)$. 
So there is some $S' \in \mc (H)$ with $S'\subseteq S$.
Since the choice of $S$ was arbitrary, we may indeed conclude that $G \entails \con H$, as required.
		\end{proof}
	
	The above result shows that evaluation essentially yields a Tarskian-style semantics for entailment by reduction to classical material implication.

\subsection{Some properties of evaluation and entailment}
\label{subsect:det-and-tot}
As we saw in the previous subsection, there are graphs that compute nonfunctional relations, for instance the $\pfour$.
However we also have examples of graphs that are deterministic but not total (i.e.\ partial Boolean functions), and total but not deterministic (i.e.\ Boolean multifunctions), as we see in the following two examples.


		\begin{example}
			[Determinism of the Bull]
			\label{ex:bull-deterministic}
			While $\pfour$ is not deterministic, it has an extension that is deterministic by adding a `settling' node:
			
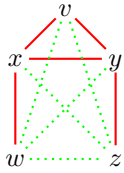
\[
			\begin{tikzpicture}
			\node[vertex] (v0) at (0.6667,2) {$v$};
			\node[vertex]  (v1) at (0, 1.3333) {$x$};
			\node[vertex]  (v2) at (1.3333, 1.3333) {$y$};
			\node[vertex] (v3) at (0, 0) {$w$};
			\node[vertex]  (v4) at (1.3333, 0) {$z$};
			\draw[r]  (v1) -- (v2) ;
			\draw[r]  (v1) -- (v3);
			\draw[g] (v1) -- (v4);
			\draw[g]  (v2) -- (v3);
			\draw[r] (v2) -- (v4);
			\draw[g]  (v3) -- (v4);
			
			\draw[r] (v0) -- (v1);
			\draw[r] (v0) -- (v2);
			\draw[g] (v0) -- (v3);
			\draw[g] (v0) -- (v4);
			\end{tikzpicture}
			\]
			The new node $v$ is called the `nose of the bull', which is particularly important in so-called `prime graphs', that we discuss in Section~\ref{sect:moddecomp}.  

The Bull computes a deterministic relation, in the sense that it never evaluates to both $0$ and $1$. We leave it to the reader to verify the cases, but this fact is also an immediate consequence of Proposition~\ref{prop:det-iff-cis} later in this section.

The Bull does not compute a total relation since, by taking the assignment $\{v,x,y\}$.
		\end{example}

		\begin{example}
			[Totality of the 5-cycle]
			\label{ex:c5-total}	
			Let us consider the $5$-cycle:
			\[
			\FiveGraphCirc{v,w,x,y,z}rggrrggrgr
			\]
			Notice that $C_5 $ has the special property that it is isomorphic to its dual, i.e.\ $C_5 \cong \compl C_5$.
			Its maximal cliques and maximal stable sets are just pairs.
			
			In fact, the $C_5 $ computes a total Boolean relation.
			To see this, let us consider the possible assignments $X$: if $|X| \leq 2$ then there is always a maximal stable set disjoint from $X$, whereas if $|X|\geq 3$ then there is always a maximal clique contained in $X$.
			
			However, the $C_5$ is not deterministic: for instance, writing $X = \{v,x,z\}$, we have that $C_5 (X,1)$ since $X$ contains the maximal clique $\{v,z\} $, but also $C_5(X,0)$ since $X$ is disjoint from the maximal stable set $\{w,y\}$.
		\end{example}
		
		We point out that the relational properties of totality and determinism behave well under duality:
		\begin{observation}
		 $G$ is deterministic (or total) if and only if $\compl G$ is deterministic (respectively total).
		\end{observation}
		This follows immediately from Remark~\ref{rmk:duality} by dualising assignments.

		One interesting observation is that the class of deterministic graphs coincides with a well-studied class in graph theory, namely the \emph{CIS} graphs, where every maximal clique and maximal stable set intersect \cite{CIS-report}.
		
		\begin{definition}
			[CIS graphs]
			A graph $G$ is \emph{CIS} if $\forall S \in \mc(G). \forall S \in \ms(G). S \cap T \neq \emptyset$
		\end{definition}
		
		Notice that cliques and stable sets may intersect at most once, so the above definition is equivalent to requiring that every maximal clique and maximal stable set have singleton intersection.
		
		It is not hard to see the following:
		\begin{proposition}
			\label{prop:det-iff-cis}
			A graph $G$ is deterministic if and only if it is CIS.
		\end{proposition}
	\begin{proof}
		For the left-right implication we prove the contrapositive. Suppose $S\in \mc (G) $ and $T \in \ms(G)$ such that $S\cap T = \emptyset$. Then $G(S,0)$ since $S\subseteq S$ and also $G(S,1)$ since $S\cap T = \emptyset$.
		
		For the right-left implication we also prove the contrapositive. Suppose $G(X,0)$ and $G(X,1)$ for some assignment $X$. Thus there is some $S \in \mc (G)$ and $T \in \ms(G)$ with $S \subseteq X$ and $ T \cap X = \emptyset$, and so $S \cap T = \emptyset$.
	\end{proof}

		As we saw in Example~\ref{ex:bull-deterministic}, adding a `settling node' was a way to force the CIS property for the $P_4$. 
		However, every $P_{4}$ configuration in a graph being `settled' does not suffice to conclude that the graph is CIS. Indeed, consider the following example from \cite{ABG:CIStechreport06}:
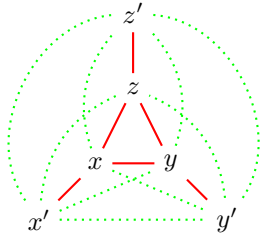
		\begin{align*}
		\begin{minipage}{.4\textwidth}
		\begin{tikzpicture}
		\tikzstyle{g} = [draw,thick,dotted,-,green]
		\tikzstyle{r} = [draw,thick,-,red]
		\node (v0) at (-.25,-.25) {$x'$};
		\node (v1) at (0.5,0.5) {$x$};
		\node (v2) at (1.5,0.5) {$y$};
		\node (v3) at (2.25,-.25) {$y'$};
		\node (v4) at (1,1.5) {$z$};
		\node (v5) at (1,2.5) {$z'$};
		\draw[r] (v0) -- (v1);
		\draw[r] (v1) -- (v4);
		\draw[r] (v2) -- (v1);
		\draw[r] (v2) -- (v3);
		\draw[r] (v4) -- (v5);
		\draw[r] (v2) -- (v4);
		\draw[g] (v1) edge (v3);
		\draw[g] (v0) edge (v2);
		\draw[g] (v0) edge (v3);
		\draw[g, bend left = 30] (v0) edge (v4); 
		\draw[g, bend left = 60] (v0) edge (v5);
		\draw[g, bend right = 30] (v3) edge (v4);
		\draw[g, bend right = 60] (v3) edge (v5);
		\draw[g, bend right = 30] (v2) edge (v5);
		\draw[g, bend left = 30] (v1) edge (v5);
		\end{tikzpicture}
		\end{minipage}
		\end{align*}
		This graph contains three $P_{4}$ configurations, all of which are `settled', but $\set{x,y,z}$ is a maximal clique which is disjoint from the maximal stable set $\set{x',y', z'}$.
		Indeed, characterising CIS graphs is not easy; it is an open problem whether they can be recognised in polynomial time \cite{ABG:CIStechreport06} pp.~2 (though obviously CIS is in $\coNP$). Therefore, deciding whether a graph computes a deterministic relation is {a priori} computationally difficult.

		It seems more difficult to establish a characterisation of the total graphs by structural graph theoretic properties. 
		This is because the definition of totality a priori is more logically complex that that of determinism, a priori a $\Pi^p_2$-property.
		We do not know of any better upper bound for recognising totality.

%
%

Determinism and totality of a graph also has an effect on when the two notions of entailment, $\entails \con$ and $\entails \dis$ coincide:
		\begin{proposition}
		\label{prop:det-tot-entailment}
			If $G$ is deterministic and $H$ is total, then:
			\begin{equation}
			\label{eqn:dis-imp-con}
			\text{if} \quad
			G \entails{\dis} H \quad \text{then}\quad G \entails \con H
			\end{equation}
			\begin{equation}
			\label{eqn:con-imp-dis}
			\text{if} \quad
			H \entails{\con} G \quad \text{then} \quad H \entails \dis G
			\end{equation}
		\end{proposition}
		\begin{proof}
			We rely on the adequacy of our relational semantics, Theorem~\ref{thm:sound-complete}.
			For \eqref{eqn:dis-imp-con} we proceed as follows:
			\[
			\begin{array}{rcll}
			G \entails \dis H \quad &\therefore & \quad \text{if } H(X,0) \text{ then }G(X,0) \quad & \text{by adequacy}\\
			\quad & \therefore & \quad \text{if }\neg G(X,0) \text{ then } \neg H(X,0) \quad & \text{by contraposition}\\
			\quad & \therefore & \quad \text{if } G(X,1) \text{ then } \neg H(X,0) \quad & \text{by determinism of $G$}\\
			\quad & \therefore & \quad \text{if }G(X,1) \text{ then } H(X,1) \quad & \text{by totality of $H$} \\
			\quad & \therefore & \quad G \entails \con H \quad & \text{by adequacy}
			\end{array}
			\]
			
			The proof of \eqref{eqn:con-imp-dis} is similar.
		\end{proof}

		Of course, an immediate consequence of the above result is that, on the class of deterministic and total graphs, $\entails{\con} = \entails \dis$.
		This is subsumed by our later characterisation of the graphs computing Boolean functions as just the $\pfour$-free ones, i.e.\ the relation webs.	
		
		It would be interesting to establish some sort of converses to the above result, determining whether a graph is total or deterministic based on how it interacts with respect to the two notions of entailment.
		We leave this for future work.

	\subsection{Case study: finer interpolation of linear inferences}
	\label{subsect:case-study}
	At least one advantage of Boolean Graph Logic is that we have a finer notion of entailment that in turn admits a richer notion of `proof'. We discuss a particular example in this subsection.
	
From \cite{Das:13:Rewritin:uq} we have that the following is a valid implication:
\begin{equation}
\label{eqn:10varinf}
\begin{array}{rl}
& (u \dis (v \con v')) \con (w \dis x) \con (w'\dis x') \con ((y\con y')\dis z) \\
\Rightarrow &
(u \con (w \dis y)) \dis (w' \con y') \dis (v' \con x') \dis ((v \dis x)\con z)
\end{array}
\end{equation}
In terms of relation webs, the implication above corresponds to the following entailments on graphs,
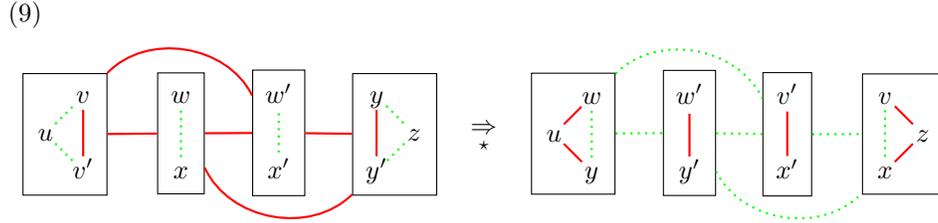
\begin{equation}
\label{eqn:10var-inf-webs}
\raisebox{-0.5\height}{\begin{tikzpicture}
	\node[vertex] (u) at (0.5,0) {$u$};
	\node[vertex] (w) at (1,0.5) {$v$};
	\node[vertex] (y) at (1,-0.5) {$v'$};
	
	\draw[g] (u) -- (w);
	\draw[g] (u) -- (y);
	\draw[r] (w) -- (y);
	\node[draw,fit=(u) (w) (y) ] (M1) {};
	
	\node[vertex] (w') at (2.3,0.5) {$w$};
	\node[vertex] (y') at (2.3,-0.5) {$x$};
	\draw[g] (w') -- (y');
	
	\node[draw,fit=(w') (y') ] (M2) {};
	
	\node[vertex] (v') at (3.6,0.5) {$w'$};
	\node[vertex] (x') at (3.6,-0.5) {$x'$};
	\draw[g] (v') -- (x');
	
	\node[draw,fit=(v') (x') ] (M3) {};
	
	\node[vertex] (v) at (4.9,0.5) {$y$};
	\node[vertex] (x) at (4.9,-0.5) {$y'$};
	\node[vertex] (z) at (5.4,0) {$z$};
	
	\draw[r] (v) -- (x);
	\draw[g] (v) --(z);
	\draw[g] (x) --(z);
	
	\node[draw,fit=(v) (x) (z) ] (M4) {};
	
	\draw[r] (M1) -- (M2);
	\draw[r] (M2) -- (M3);
	\draw[r] (M1) to[bend left=55] (M3);
	\draw[r] (M3) -- (M4);
	\draw[r,bend right=55] (M2) edge (M4);
	\end{tikzpicture}}
\quad \entails{\star} \quad 
\raisebox{-0.5\height}{\begin{tikzpicture}
	\node[vertex] (u) at (0.5,0) {$u$};
	\node[vertex] (w) at (1,0.5) {$w$};
	\node[vertex] (y) at (1,-0.5) {$y$};
	
	\draw[r] (u) -- (w);
	\draw[r] (u) -- (y);
	\draw[g] (w) -- (y);
	\node[draw,fit=(u) (w) (y) ] (M1) {};
	
	\node[vertex] (w') at (2.3,0.5) {$w'$};
	\node[vertex] (y') at (2.3,-0.5) {$y'$};
	\draw[r] (w') -- (y');
	
	\node[draw,fit=(w') (y') ] (M2) {};
	
	\node[vertex] (v') at (3.6,0.5) {$v'$};
	\node[vertex] (x') at (3.6,-0.5) {$x'$};
	\draw[r] (v') -- (x');
	
	\node[draw,fit=(v') (x') ] (M3) {};
	
	\node[vertex] (v) at (4.9,0.5) {$v$};
	\node[vertex] (x) at (4.9,-0.5) {$x$};
	\node[vertex] (z) at (5.4,0) {$z$};
	
	\draw[g] (v) -- (x);
	\draw[r] (v) --(z);
	\draw[r] (x) --(z);
	
	\node[draw,fit=(v) (x) (z) ] (M4) {};
	
	\draw[g] (M1) -- (M2);
	\draw[g] (M2) -- (M3);
	\draw[g] (M1) to[bend left=55] (M3);
	\draw[g] (M3) -- (M4);
	\draw[g,bend right=55] (M2) edge (M4);
	\end{tikzpicture}}
\end{equation}
for $\star \in \{\dis,\con \}$.

\todo{Is the next sentence true?}
This is a minimal linear inference in the sense that no linear formula (over the same variables) simultaneously implies the RHS and is implied by the LHS. However, in our graph-theoretic setting, we may indeed `interpolate' this inference thanks to the following intermediate graph:
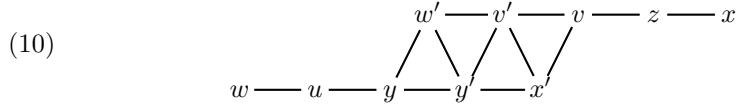
\begin{equation}
\label{eqn:interp-graph}
\raisebox{-0.5\height}{\begin{tikzpicture}
	\node[vertex] (w) at (0,0) {$w$};
	\node[vertex] (u) at (1,0) {$u$};
	\node[vertex] (y) at (2,0) {$y$};
	\node[vertex] (w') at (2.5,1) {$w'$};
	\node[vertex] (y') at (3,0) {$y'$};
	\node[vertex] (v') at (3.5,1) {$v'$};
	\node[vertex] (x') at (4,0) {$x'$};
	\node[vertex] (v) at (4.5,1) {$v$};
	\node[vertex] (z) at (5.5,1) {$z$};
	\node[vertex] (x) at (6.5,1) {$x$};
	
	\draw[b] (w) -- (u);
	\draw[b] (u) -- (y);
	\draw[b] (y) -- (y');
	\draw[b] (y) -- (w');
	\draw[b] (w') -- (y');
	\draw[b] (y') -- (v');
	\draw[b] (w') -- (v');
	\draw[b] (y') -- (x');
	\draw[b] (v') -- (x');
	\draw[b] (v') -- (v);
	\draw[b] (x') -- (v);
	\draw[b] (v) -- (z);
	\draw[b] (z) -- (x);
	\end{tikzpicture}}
\end{equation}
Note that we have indicated only edges, not non-edges, with black lines, for clarity.

\eqref{eqn:interp-graph} contains several $\pfour$s, e.g.\ $\{w,u,y,y'\} $ and $\{y,y',v',v\}$. 
It does not compute a deterministic relation since $\{w,y,v',x \} \in \ms \eqref{eqn:interp-graph}$ is disjoint from $\{v,z\} \in \mc \eqref{eqn:interp-graph}$.
On the other hand:
\begin{claim}
	\label{claim:interpolant-total}
\eqref{eqn:interp-graph} computes a total relation.
\end{claim}
\begin{proof}
	Let $X$ be an assignment such that $\neg \eqref{eqn:interp-graph}(X,1)$, and let us try to construct a maximal stable set $T$ disjoint from it.
	Note that $X$ cannot contain all of $S = \{y,w',y'\} \in \mc\eqref{eqn:interp-graph}$ and all of $S' = \{v',x,v\}\in \mc\eqref{eqn:interp-graph}$, so let $b \in S$ and $b' \in S'$ be distinct nodes with $b,b' \notin X$.
	
	Similarly, we must have that $X$ does not contain one of $w$ or $u$, say $a$, and that $X$ does not contain one of $z $ or $x$, say $c$.
	
	In fact, all possible values of $a,b,b',c$ yield a maximal stable set, which by construction is disjoint from $X$. Note in particular that, if $a=u$ then $b\neq y$, otherwise we would have $\eqref{eqn:interp-graph}(X,1)$, and if $c=z$ then $b' \neq v$, for the same reason. 
	Thus $\eqref{eqn:interp-graph}(X,0)$ as required.
	\end{proof}

As we mentioned, \eqref{eqn:interp-graph} indeed interpolates the inference from \eqref{eqn:10varinf}, in fact for both versions of entailment, as we will now show.
The entailments $\lhs \eqref{eqn:10var-inf-webs} \entails \con \eqref{eqn:interp-graph}$ and $\eqref{eqn:interp-graph} \entails \dis \rhs\eqref{eqn:10var-inf-webs}$ require quite a large case analysis, due to the high number of maximal cliques in $\lhs \eqref{eqn:10var-inf-webs}$ and the high number of maximal stable sets in $ \rhs\eqref{eqn:10var-inf-webs}$.
However the other two entailments are much easier to prove directly:

\begin{proposition}
We have that $\lhs \eqref{eqn:10var-inf-webs} \entails \dis \eqref{eqn:interp-graph}$ and $\eqref{eqn:interp-graph} \entails \con \rhs\eqref{eqn:10var-inf-webs}$.
\end{proposition}
\begin{proof}
	For $\lhs \eqref{eqn:10var-inf-webs} \entails \dis \eqref{eqn:interp-graph}$, we describe how to map each maximal stable set $T$ of \eqref{eqn:interp-graph} to a subset that is maximally stable in $\lhs\eqref{eqn:10var-inf-webs}$.
	We have the following cases:
	\begin{itemize}
		\item If $T\ni u$ then we map it to $\{u\}$.
		\item if $T\ni z$ then we map it to $\{z\}$.
		\item Otherwise $T$ must contain $w$ and $x$, so we may map it to $\{w,x\}$.
	\end{itemize}
	
For $\eqref{eqn:interp-graph} \entails \con \rhs\eqref{eqn:10var-inf-webs}$, we describe how to map each maximal clique of \eqref{eqn:interp-graph} to a subset that is a maximal clique of $\rhs\eqref{eqn:10var-inf-webs}$:
\begin{itemize}
	\item $\{w,u\}$, $\{u,y\}$, $\{v,z\}$ and $\{z,x\}$ are all mapped to themselves.
	\item $\{y,w',y \}$ and $\{w', y', v' \}$ are mapped to $\{w',y'\}$.
	\item $\{y',v',x'\}$ and $\{v',x',v\}$ are mapped to $\{v',x'\}$. \qedhere
\end{itemize}
\end{proof}

Now, since \eqref{eqn:interp-graph} computes a total relation, by Claim~\ref{claim:interpolant-total} above, we may conveniently appeal to Proposition~\ref{prop:det-tot-entailment} to immediately recover the other two entailments.

\begin{corollary}
	We have both of the following:
	\begin{itemize}
		\item $\lhs \eqref{eqn:10var-inf-webs} \entails \con \eqref{eqn:interp-graph} \entails \con \rhs\eqref{eqn:10var-inf-webs}$
		\item $\lhs \eqref{eqn:10var-inf-webs} \entails \dis \eqref{eqn:interp-graph} \entails \dis \rhs\eqref{eqn:10var-inf-webs}$
	\end{itemize}
\end{corollary}

We point out that we only used the fact that $\lhs\eqref{eqn:10varinf}$ and $\rhs\eqref{eqn:10varinf}$ are deterministic to obtain the other two entailments, not that they are total. 
It would be interesting to develop this case study further, establishing a maximal sequence of graphs interpolating \eqref{eqn:10var-inf-webs}.
This is also related to the question of finding the `minimal linear inference', cf.~\cite{sipraga}.
Such development, however, is beyond the scope of this work.

\subsection{Deterministic and total graphs are $\pfour$-free}
\label{subsect:det-tot-pfour-free}

We finish this section with a somewhat surprising result: the only graphs that are both deterministic and total (i.e.\ Boolean functions) are already $\pfour$-free, i.e.\ they are the relation webs of formulas.

\begin{thm}
	\label{thm:function-p4free}
	A graph is deterministic and total if and only if it is $ P_4 $-free.
\end{thm}

One proof of this result is by a reduction to a classical result of Gurvich:

\begin{theorem}
	[\cite{gurvich1977repetition}]
	\label{thm:gurvich}
	Suppose $f: \pow (\Var) \to \{0,1\}$ is monotone and depends on variables $V_n = x_1, \dots, x_n$.  
	$f$ is computed by a read-once formula (over $V_n$) if and only if, for every minterm $S$ of $f$ and every maxterm $T$ of $f$, $|S \cap T| = 1$.
\end{theorem}

We give the reduction of Theorem~\ref{thm:function-p4free} to Theorem~\ref{thm:gurvich} here, but a self-contained proof can be found in the appendix, Section~\ref{app:det-tot-pfourfree}.

\begin{proof}
	[Proof of Theorem~\ref{thm:function-p4free}]
	Let $G$ be a deterministic and total graph, 
	and define the following two Boolean functions:
	\[
	T_G (X) := \begin{cases}
	1 & G(X,1) \\
	0 & \text{otherwise}
	\end{cases}
	\qquad
	F_G(X) := \begin{cases}
	0 & G(X,0) \\
	1 & \text{otherwise}
	\end{cases}
	\]
	Since $G$ is deterministic and total, we have that $T_G$ and $F_G$ are actually the same Boolean function, say $g$. 
	What is more, $g$ may be written simultaneously as the following (irredundant) DNF and CNF:
	\[
	\bigvee\limits_{S \in \mc (G)} \bigwedge S \quad = \quad  T_G \quad =\quad  g \quad =  \quad F_G = \bigwedge\limits_{T \in \ms (G)} \bigvee T
	\]
	Thus the minterms of $g$ are just the maximal cliques of $G$ and the maxterms of $g$ are the maximal stable sets of $G$. 
	Since cliques and stable sets may only intersect at most once (by simplicity of the graph), and minterms and maxterms must intersect at least once (by determinism of functions), we have from Gurvich's theorem above that $g$ is computed by some read-once formula, say $A$.
	But now we must have that, indeed $\web (A) = G$, since otherwise they would have different maximal cliques and stable sets, and so compute different relations. Thus $G$ is the web of some formula and, indeed, $\pfour$-free.
\end{proof}

%
%
%
%
%
%
%

\section{Computational Complexity of Boolean Graph Logic}
\label{sect:complexity}

In this section we study the computational complexity of evaluation and entailment in BGL.
In particular we show that evaluation (to either $0$ or to $1$) is $\NP$-complete and entailment (either disjunctive or conjunctive) is $\Pi^p_2$-complete.
In contrast, for Boolean formulas evluation is $\mathbf{ALOGTIME}$-complete and entailment is $\coNP$-complete, suggesting that BGL is much more computationally rich than Boolean logic.

\subsection{Preliminaries on computational complexity}
We will assume prior knowledge of deterministic and (co-)nondeterministic Turing and oracle machines.
For a language $L$ we write $\NP(L)$ for the class of languages accepted by a nondeterministic Turing machine in polynomial time with access to an oracle for $L$.
For a class of languages $\mathcal C$ we write $\NP(\mathcal C)$ for $\bigcup\limits_{L \in \mathcal C} \NP(L)$.
From here recall that the levels of the polynomial hierarchy are defined as follows:
\begin{itemize}
	\item $\ndhier{p}{0} = \cohier{p}{0} =  \PP$.
	\item $\ndhier{p}{i+1} = \NP({\ndhier{p}{i}})$.
	\item $\cohier{p}{i+1} = \mathbf{co}\ndhier{p}{i+1}$.
\end{itemize}
Of course, $\ndhier p 1$ is just $\NP$ and $\cohier p 1 $ is just $\coNP$.

%

Let us write $\vec \exists \CNF$ for the class of true (closed) quantified Boolean formulas (QBFs) of the form $\exists \vec x . \phi$, where $\phi$ is a CNF.
Similarly we write $\vec \forall \vec \exists \CNF$ for the class of true (closed) quantified Boolean formulas (QBFs) of the form $\forall \vec x . \exists \vec y . \phi$, where $\phi$ is a CNF.
It is well-known that $\vec \exists \CNF$ is $\NP$-complete (being the same as $\SAT$), and that $\vec \forall \vec \exists \CNF$ is $\cohier p 2 $-complete (e.g.\ see \cite{schaefer2002completeness}).

%
%

\subsection{Complexity of entailment}
We show that the relations $\entails{\con}$ and $\entails{\dis}$ are complete for $\cohier{p}{2}$, i.e.\ $\coNP(\NP)$.

\begin{theorem}
	\label{thm:entailment-pip2-complete}
	$\entails{\con}$ and $\entails{\dis}$ are $\cohier{p}{2}$-complete.
\end{theorem}
\begin{proof}
	We reduce $\vec \forall \vec \exists \CNF$ to $\entails{\dis}$, whence the case for $\entails{\con}$ follows by duality.
	Fix an instance $\phi$,
	\[
	\forall \vec x . \exists \vec y . \bigwedge\limits^N_{n=1} \bigvee C_n
	\]
	where each $C_n$ is a set of literals over the variables $\vec x , \vec y$.
	Write $\phi_0$ for the matrix of $\phi$, i.e.\ $\bigwedge\limits^N_{n=1} \bigvee C_n$.
	\begin{remark}
		\label{rmk:cnf-convention}
		Without loss of generality, we assume $\vec x$ and $\vec y $ are disjoint and that each $x_i$ and $y_i$ occurs both positively and negatively in $\phi$ (otherwise replace it by $0$ or $1$ appropriately).
		Furthermore, we suppose that that each $C_n$ contains some $x_i$ and some $y_j$ (i.e.\ a universally bound variable and an existentially bound variable), either positively or negatively.
		This is because if this were not the case for some clause $C$, then it could equivalently be replaced by clauses $C\cup \{z\}$ and $C\cup \{ \neg z \}$, for any variable $z$. Such a procedure at most multiplies the size of $\phi_0$ by $2$.
	\end{remark}
	
	\todo{draw graphs?}
	
	Now we define the graphs $G = (V, E_G)$ and $H = (V, E_H )$ as follows:
	\begin{itemize}
		\item The set $V$ of vertices of both $G$ and $H$ is the set of literal occurrences in $\phi_0$.
		Formally, we write $V = \{ x_i^j, \compl x_i^k, y_i^j, \compl y_i^j \}_{i,j} $, where $j$ identifies the specific occurrence of each literal and $i,j$ range appropriately.
		\item The set $E_G$ of edges of $G$ consists of:
		\begin{itemize}
			\item an edge between any two nodes in the same clause (so that each $C_n$ is a clique); and,
			\item an edge between any two nodes of the form $y_i^j$ and $\compl y_i^k$ (i.e.\ any dual literal occurrences that are existentially bound).
		\end{itemize}
		\item $E_H$ consists only of edges between nodes of the form $x_i^j$ and $\compl x_i^k$ (i.e.\ any dual literal occurrences that are universally bound).
	\end{itemize}

	\smallskip 
	
	We claim that $G \entails{\dis} H $ if and only if $\phi$ is true, as required. 
	First suppose that $G\entails\dis H$ and let $X \subseteq \vec x$ be an assignment to the $x_i$'s.
	Let $X_1 \subseteq V$ identify the true literal occurrences under $X$, i.e.:
	\begin{equation}
	\label{eqn:true-literals-of-X}
	X_1 \quad := \quad \{ x_i^j : x_i \in X \} \cup \{  \compl x_i^k : x_i \notin X \} 
	\end{equation}
	Now define $T$ to be the set of all $\vec y$-literal occurrences and the true $\vec x$-literal occurrences under $X$, i.e.:
	\[
	T = \{ y_i^j \}_{i,j} \cup \{ y_i^k \}_{i,k} \cup X_1
	\]
	
	Notice that $T$ is a stable set in $H$ by definition of $E_H$ and the fact that $X_1$ does not contain dual literals.
	Furthermore it is maximally stable since the only remaining nodes are $\vec x$-literals that are false under $X$, and so have an edge to some true literal occurrence in $X_1$ (cf.~Remark~\ref{rmk:cnf-convention}).
	
	Thus, since $G \entails{\dis} H$, there is a set $T' \subseteq T$ that is maximally stable in $G$.
	Now, notice that,
	\begin{enumerate}
		\item\label{item:ms-intersects-clauses} By maximality, $T'$ must intersect every $C_n$, since there is some $x^j_i$ or $\compl x^k_i$ in each $C_n$, by Remark~\ref{rmk:cnf-convention} and by the definition of $E_G$; and,
		\item\label{item:ms-consistent} $T'$ cannot contain any dual pair $y^l_j$ and $\compl y^m_j$, by the definition of $E_G$.
	\end{enumerate} 
	Thus $T'$ induces a consistent assignment $Y \subseteq \vec y$ to the $y_i$s (just take the set of $\vec y$-literals in $T'$), by \ref{item:ms-consistent},  that further ensures that there is a true literal in each $C_n$ under $X,Y$, by \ref{item:ms-intersects-clauses}, as required.
	
	Conversely, suppose that $\phi $ is true and let $T \in \ms(H)$.
	By the definition of $E_H$,
	$T$ induces a consistent assignment $X \subseteq \vec x$ to the $x_i$s in the natural way, so let $Y \subseteq \vec y $ be an assignment to the $y_i$s obtained by the truth of $\phi$, i.e.\ such that $\phi_0 (X , Y)$ is true.
	Let $X_1$ be defined as above in \eqref{eqn:true-literals-of-X} and $Y_1$ be defined similarly, i.e.\ $X_1 \subseteq V$ and $Y_1 \subseteq V$ consist of the true $\vec x$-literal occurrences and $\vec y$-liter occurrences, respectively, under $X$ and $Y$, respectively. Note that $X_1 \cup Y_1$ must intersect each $C_n$, since $\phi_0(X,Y)$ is true.
	
	Now we are almost able to define an appropriate subset of $T $ that is maximally stable in $G$, but for one technicality: $ X_1 \cup Y_1$ could intersect some clause $C_n$ \emph{twice}, i.e.\ there could be two true literals in $C_n$.\footnote{We could have avoided this issue by working with generalisations of exactly-one-in-three-$\SAT$, but this is not standard and is beyond the scope of this work.}
	For this reason, we rather set $T'$ to be an arbitrary subset of $X_1 \cup Y_1$ that intersects each $C_n$ \emph{exactly} once:\footnote{We assume there is some fixed established ordering of the literals.}
	\[
	T' \quad : = \quad \{ \text{least literal in $(X_1 \cup Y_1) \cap C_n$ }  : 0< n\leq N    \}
	\]
	
	Now we have the following:
	\begin{itemize}
		\item $T'$ is stable in $G$, since it only contains one node in each $C_n$ and is consistent with the truth assignment $Y$ to the $y_i$s, cf.~the definition of $E_G$; and,
		\item Furthermore $T'$ is maximally stable in $G$, since it already intersects each $C_n$.
	\end{itemize}
	We have $T'\subseteq T$ (since $T$ contains, in particular, all $\vec y$-literal occurrences), so we indeed have that $G \entails \dis H$, as required.
\end{proof}

\begin{remark}
	In \cite{Stra:08:Extensio:kk} Stra\ss burger showed that the linear fragment of Boolean logic is $\coNP$-complete, by rewriting every Boolean tautology $\forall \vec x . \phi$ as a linear inference $\forall \vec x . (\phi \imp \psi)$, where $\phi$ and $\psi$ are monotone and linear.
	It would be natural to try use this approach here to show $\cohier{p}{2}$-completeness of entailment, but we point out that such an approach could not work, unless polynomial hierarchy collapses to $\coNP$. 
	This is because the set of true $\vec \forall \vec \exists$-formulas whose matrices are even monotone implications, let alone linear, is already in $\coNP$.
	To see this, suppose otherwise, and consider a closed formula:
	\begin{equation}
	\label{eqn:mon-imp}
	\forall \vec x . \exists \vec y . ( \phi(\vec x, \vec y) \imp \psi (\vec x , \vec y)  )
	\end{equation}
	where $\phi$ and $\psi$ are monotone.
	We have the following:
	\[
	\begin{array}{rcll}
	\exists \vec y . (\phi (\vec x , \vec y) \imp \psi (\vec x, \vec y)) & \iff & \forall \vec y \phi (\vec x , \vec y) \imp \exists \vec y \phi(\vec x , \vec y) & \text{by De Morgan equivalences} \\
	& \iff & \phi(\vec x , \vec 0 ) \imp \psi (\vec x , \vec 1) & \text{by monotonicity}
	\end{array}
	\]
	Thus we have reduced the truth of \eqref{eqn:mon-imp} to the truth of $\forall \vec x . (\phi(\vec x , \vec 0 ) \imp \psi (\vec x , \vec 1))$, which is, of course, in $\coNP$.
\end{remark}

\subsection{Complexity of evaluation}
By adequacy, Theorem~\ref{thm:completeness}, and unwinding the definition of evaluation, we already have as a corollary of Theorem~\ref{thm:entailment-pip2-complete} that there cannot be a polynomial-time algorithm for evaluation, unless the polynomial-hierarchy collapses to $\coNP$. In fact we even have that there cannot be such an algorithm in $\coNP\cap\NP$, for the same reason.

In this section we go further and show that evaluation is in fact $\NP$-complete. 
This suggests that there is no `local' vertex-contraction style procedure for evaluation, unless $\PP=\NP$.

\begin{theorem}
	The graph evaluation problems, i.e.\ $G(X,0)$ and $G(X,1)$, are $\NP$-complete.
\end{theorem}
\begin{proof}
	We reduce $\vec{\exists}\CNF$ (i.e.\ $\SAT$) to graph evaluation, namely the problem $G(X,0)$.
	(Again, the case of $G(X,1)$ is obtained by duality).

	Let $ \phi$ be a formula $  \exists \vec x . \bigwedge\limits^N_{n=1} \bigvee C_n$ where each $C_n$ is a set of literals over the variables $\vec x$. 
	Write $\phi_0$ for the matrix of $\phi$, i.e.\ $\bigwedge\limits^N_{n=1} \bigvee C_n$.
	
	\begin{remark}
		\label{rmk:no-dual-variables}
		We assume without loss of generality that each $C_n$ contains at most one of $x$ and $\compl x$, for any variable $x$ (otherwise just delete the clause).
	\end{remark}
	
	Let $L$ denote the set of literal occurrences in $\phi_0$ and define a set of fresh nodes $C=\{ c_1, \dots, c_N \}$.
	We define the graph $G=(V,E)$ as follows:
	\begin{itemize}
		\item The set $V$ of nodes of $G$ is $L \cup C$.
		\item The set $E$ of edges of $G$ consists of:
		\begin{itemize}
			\item an edge between any two literal occurrences in the same clause; and,
			\item an edge between any occurrence of $x_i$ and any occurrence of $\compl x_i$; and,
			\item an edge between any literal occurrence in a clause $C_n$ and $c_n$.
		\end{itemize}
	\end{itemize}

	We claim that $G(C,0)$ if and only if $\phi$ is true, as required.
	First, suppose that $G(C,0)$ and let $T \in\ms(G)$ be disjoint from $C$, by definition of $G(C,0)$.
	We have the following:
	\begin{itemize}
		\item $T$ intersects each $C_n$. Otherwise, it would contain $c_n$ and thus intersect $C$.
		\item $T$ is \emph{consistent}, i.e.\ if it contains an occurrence of $x$ it does not contain an occurrence of $\bar x$, by definition of $E$.
	\end{itemize}
	Thus $T$ induces an assignment $X \subseteq \vec x$ in the natural way such that $\phi_0(X)$ is true.
	
	Conversely, suppose that $\phi$ is true and let $X \subseteq \vec x$ be a satisfying assignment for $\phi_0$.
	The set of true literal occurrences in $G$ under $X$ (defined just like $X_1$ was defined in \eqref{eqn:true-literals-of-X})
	almost serves as an appropriate maximal stable set but for the technicality, as before, that it may intersect some clause more than once.
	Again, we avoid this issue by making arbitrary choices.
	Define:
	\[
	T \quad := \quad \{ \text{least true literal under $X$ in $C_n$} : 0<n\leq N \}
	\]
	We have the following:
	\begin{itemize}
		\item $T$ is stable, by Remark~\ref{rmk:no-dual-variables} and consistency of $X$;
		\item $T$ intersects each $C_n$, since $X$ was a satisfying assignment for $\phi_0$, and so $T$ is furthermore maximally stable by definition of $E$; and,
		\item $T$ does not intersect $C$, by construction.
	\end{itemize}
	Thus $T \in \ms(G)$ such that $T\cap C = \emptyset$, as required.
\end{proof}

%
%
%
		
		\section{Modular decomposition and an algorithm for evaluation}
		\label{sect:moddecomp}

		Despite evaluation being $\NP$-complete, we can still define a ``recursive'' algorithm for it based on known graph decompositions.
		While this algorithm does not operate in polynomial time, it does allow us to reduce evaluation, as well as determining whether a graph is deterministic or total, to the so-called ``prime'' graphs.

		\subsection{Modules: `zooming' out of graphs}
The notion of a module generalises the notion of a formula context to arbitrary graphs.		

		\begin{definition}
			[Modules]
			For a graph $G = (V,E)$ a \emph{module} is a set $M \subseteq V$ such that every element of $M$ has the same neighbourhood outside $M$ in $G$.
			I.e.\ 
			$\forall x,y \in M . \forall z \in V \setminus M . (\{x,z\} \in E \text{ iff } \{y,z \} \in E)$.
			
			The sets $ V, \emptyset  $ and $ \set{x} $, for $ x \in V $, are always modules and are known as the \e{trivial modules}.  
			Any module $ M \subsetneq V $ is a \e{proper} module.

		\end{definition}
	
	Rephrasing the above definition visually, $M$ being a module means that, for any $z \notin M$, either:
	\begin{itemize}
		\item $\forall m \in M. \redge[]{m}{y}$; or	
		\item $\forall m \in M. \gedge{m}{y} $.
	\end{itemize}
As a notational convention, given a graph $G$, for two sets of nodes $ X,Y $ we write $ \redge{X}{Y}$ (or $ \gedge{X}{Y}  $) to express that for all $x \in X $ and all $ y \in Y $, $ \redge{x}{y} $ (respectively, $ \gedge{x}{y} $) in $G$.

\begin{observation}\label{obs:modulesdisjoint}
	For disjoint modules $ M,N $, we have either $ \redge{M,N} $ or $ \gedge{M,N} $
\end{observation}

In this way, modules allow us to `zoom out' and see graphs as compositions of smaller graphs.
This is similar to how we displayed the graphs in the case study of Section~\ref{subsect:case-study}.
Let us elaborate on this idea more formally.

\begin{definition}
	[Quotients]
	Fix a graph $G=(V,E)$ and let $P\subseteq \pow (V)$ be a partition of $V$ into (nonempty) modules (called a \emph{modular partition}).
	The \emph{quotient graph} $G/P$ is defined as follows:
\begin{itemize}
	\item The set of vertices of $G/P$ is just $P$.
	\item $\{M,N\}$ is an edge of $G/P$ just if $\redge M N $ in $G$.
\end{itemize}
\end{definition}

Notice that we deliberately use modules themselves as nodes in a quotient graph. This allows us to freely switch between consideration of the entire graph and just its quotient when seeing graphs visually.
For the same reason, we will often only speak about quotients up to isomorphism.

\begin{example}
	\label{ex:modular-partition-non-maximal}
Consider the following graph $G$, written with only edges indicated:
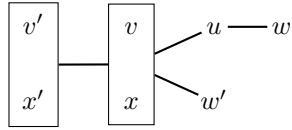
\[
{\begin{tikzpicture}
	\node[vertex] (v') at (0,0) {$v'$};
	\node[vertex] (x') at (0,-1) {$x'$};

	\node[draw,fit=(v') (x')] (M1) {};
	
	\node[vertex] (v) at (1.3,0) {$v$};
	\node[vertex] (x) at (1.3,-1) {$x$};

	\node[draw,fit=(v) (x) ] (M2) {};
	
	\node[vertex] (u) at (2.4,0) {$u$};
		\node[vertex] (w') at (2.4,-1) {$w'$};
	\node[vertex] (w) at (3.3,0) {$w$};
	
	\draw[b] (M1) -- (M2);
	\draw[b] (M2) -- (u);
		\draw[b] (M2) -- (w');
	\draw[b] (u) -- (w);

	\end{tikzpicture}}
\]
Here we have identified two nontrivial modules, $M = \{v',x'\}$ and $N= \{v,x\}$ whose elements have the same adjacencies.
From this presentation, we may isolate a particular partition of the vertices into modules:
\[
P = \{ M, N, \{u\}, \{w'\}, \{w\} \}
\]
The graph $G/P$ is thus the following:
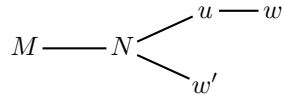
\[
{\begin{tikzpicture}

	\node[vertex] (M1) at (0,-0.5) {$M$};
	

	\node[vertex] (M2) at (1.3,-0.5) {$N$};
	
	\node[vertex] (u) at (2.4,0) {$u$};
	\node[vertex] (w') at (2.4,-1) {$w'$};
	\node[vertex] (w) at (3.3,0) {$w$};
	
	\draw[b] (M1) -- (M2);
	\draw[b] (M2) -- (u);
	\draw[b] (M2) -- (w');
	\draw[b] (u) -- (w);
	
	\end{tikzpicture}}
\]
Notice that, as an overloading of notation, we may use the same diagram above as a representation of $G$ itself, identifying $M$ and $N$ with $G\restr M$ and $G\restr N$ respectively.
\end{example}

Notice that the module $M$ in the previous example is not maximally proper, as it may be extended by $w'$.
The point of modular decompositions we later define is to take quotients as finely as possible, recursively expressing it as a `graph of graphs'.
However, some graphs cannot be simplified in this way, and these form the critical points of modular decomposition.

\begin{definition}
	[Prime graphs]
		Let $G$ be a graph of size at least $3$. If every module in $ G $ is trivial, we say that $ G $ is a \e{prime} graph. 
\end{definition}
Prime graphs have been studied extensively, see for example \cite{COURNIER199861} \cite{ILLE201476} \cite{habib2010survey}. 

\subsection{Modular decomposition of a graph}
		
We have the following natural algebraic properties of modules:
		\begin{proposition}
			[Algebra of modules]
				Let $ G $ be a graph with modules $M$ and $N$.
			\begin{enumerate}
				\item $ M \cap N $ is a module.
				\item If $ M\cap N \neq \emptyset$ then $ M \cup N $ is a module.
			\end{enumerate}	
			
		\end{proposition}

As a result of this algebraic structure, we have the following well-known decomposition result for graphs:

\begin{proposition}
	[Modular decomposition of a graph, \cite{gallai1967transitiv}]
	\label{prop:mod-decomp}
	For every nonempty graph $G$, we have exactly one of the following:
	\begin{enumerate}
		\item\label{item:singleton-graph} $G$ is a singleton graph $\{x\}$.
		\item\label{item:parallel-comp} $G$ is disconnected.
		\item\label{item:series-comp} $G$ is co-disconnected (i.e.\ $\compl G$ is disconnected).
		\item\label{item:biconnected} $V(G)$ is partitioned by its maximal proper modules.
	\end{enumerate}
\end{proposition}

This motivates the following definition:
\begin{definition}
	[Prime quotient]
	Given a biconnected graph $G = (V,E)$, its \emph{prime quotient}, written $P_G$, is the set of its maximal proper modules.
\end{definition}

Under Proposition~\ref{prop:mod-decomp}, we have that $G/P_G$ (often simply written $P_G$, as abuse) is a graph with edges $\{M,N\}$ just if $\redge M N $ in $G$.
As the name suggests, the prime quotient is indeed a prime graph:

%
%
%
%
%
%
%

	\begin{fact}
		The prime quotient of a bi-connected graph is a prime graph.
	\end{fact}
	
	\begin{example}
		\label{ex:modular-partition-revisited}
Revisiting Example~\ref{ex:modular-partition-non-maximal}, notice that the module $M$ is not maximally proper, since it may be extended by $w'$. Let us call the resulting module $M'$.
The prime quotient of $G$ is $P_G = \{M',N,\{u\},\{w\} \}$, so that $G/P_G$ is actually a $\pfour$:
\[
\FourGraphPath{M',N,u,w}bwwbwb
\]
The $\pfour$ is the smallest prime graph.
	\end{example}
%
%
%
%
%
%

%
%
%

			We are ready to define the modular decomposition tree of a graph.
			
		\begin{definition}
			[Decomposition tree]
			\label{dfn:dec-tree}
			We define the \emph{decomposition tree} of a graph $G$, written $\tree(G)$ by induction on its size, under the classification of Proposition~\ref{prop:mod-decomp}:
			\begin{enumerate}
				\item\label{item:md-singleton} If $G$ has just one node, i.e.\ is $(\{x\}, \emptyset)$, then $\tree(G) := G$.
				\item\label{item:md-dis} If $G$ is disconnected with connected components $G_1, \dots , G_n$, then $\tree (G) $ is:
				\[
				\begin{tikzpicture}
				\node{$\dis$}
				child {
					node (a) {$\tree(G_1) \quad$}
				}
				child {
					node (b) {$\quad \tree(G_n)$}
				};
				\path (a) -- (b) node [ midway] {$\cdots$};
				\end{tikzpicture}
				\]
				We write $\tree(G) = \vee (\tree (G_1), \dots , \tree(G_n))$ as a more compact notation.
				\item\label{item:md-con} If $\compl{G}$ is disconnected with connected components $\compl G_1 , \dots \compl G_n$, then $\tree(G)$ is:
				\[
				\begin{tikzpicture}
				\node{$\con$}
				child {
					node (a) {$\tree(G_1) \quad$}
				}
				child {
					node (b) {$\quad \tree(G_n)$}
				};
				\path (a) -- (b) node [ midway] {$\cdots$};
				\end{tikzpicture}
				\]
				We write $\tree(G) = \wedge (\tree (G_1), \dots , \tree(G_n))$ as a more compact notation.
				\item\label{item:md-prime} Otherwise let $P_G = \{M_1, \dots, M_n\}$ and define $\tree(G)$ as,
				
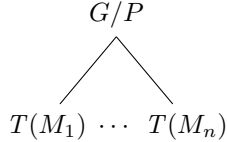
\[
				\begin{tikzpicture}
				\node{$G/P$}
				child {
					node (a) {$\tree(M_1) \quad$}
				}
				child {
					node (b) {$\quad \tree(M_n)$}
				};
				\path (a) -- (b) node [ midway] {$\cdots$};
				\end{tikzpicture}
				\]
				where we identify each $M_i$ with the corresponding induced subgraph $G\restr M_i$.
				We write $\tree(G) = (G/P) (\tree(M_1),\dots, \tree(M_n) )$ as more compact notation.
			\end{enumerate}
		\end{definition}

\begin{remark}
	Notice that, identifying formulas with their formula trees, the mapping $\tree (\cdot)$ on cographs is precisely the inverse of the mapping $\web (\cdot)$, mapping a formula to its relation web, up to associativity and commutativity of $\vee$ and $\wedge$.
\end{remark}

Let us see some examples of decomposition trees in action.

\begin{example}
Revisiting Examples~\ref{ex:modular-partition-non-maximal} and \ref{ex:modular-partition-revisited}, we have that $T(G)$ is as follows:
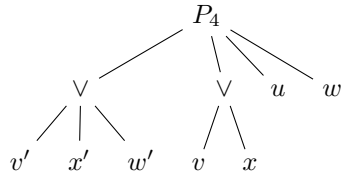
\[
\begin{forest}
[
$\pfour$ [
$\vee$ [ $v'$] [$x'$] [$w'$ ]
]
[
$\vee$ [$v$] [$x$]
]
[$u$]
[$w$]
]
\end{forest}
\]
Notice that we have simply written $\pfour$ as the root of $\tree(G)$ rather than the proper isomorphic quotient graph, but this causes no ambiguity here.

The compact notation for $\tree(G)$ is $\pfour ( \vee(v',x',w') , \vee (v,x), u , w )$.
\end{example}

	\begin{ex}
		Let $ G $ be the following graph:
		\[ \SevenMod{c,d,e,f,g,a,b}rggggrgggrrgrgggggggr \]
		Notice that we may equivalently write $G$ in the following way,
		
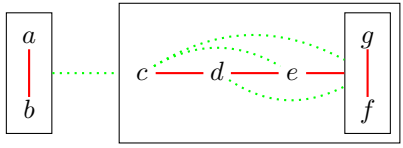
\[
		\begin{tikzpicture}
		\node[vertex] (a) at (-0.5,0) {$a$};
		\node[vertex] (b) at (-0.5,-1) {$b$};
		
		\draw[r] (a) -- (b);
		
		\node[draw, fit=(a) (b) ] (M1) {};
		
		\node[vertex] (c) at (1,-0.5) {$c$};
			\node[vertex] (d) at (2,-0.5) {$d$};
			\node[vertex] (e) at (3,-0.5) {$e$};
			
			\node[vertex] (f) at (4,-1) {$f$};
			\node[vertex] (g) at (4,0) {$g$};
			\draw[r] (f) -- (g);
			
			\node[draw, fit=(f) (g) ] (M3) {};
			
			\draw[r] (c) -- (d);
			\draw[r] (d) -- (e);
			\draw[r] (e) -- (M3);
			
			\draw[g, bend left] (c) edge (e);
			\draw[g, bend right] (d) edge (M3);
			\draw[g, bend left] (c) edge (M3);
			
			\node[draw,fit=(c) (d) (e) (M3)] (M2) {};
			
			\draw[g] (M1) -- (M2);
		\end{tikzpicture}
		\]
%
%
%
and thus $\tree(G)$ has the following form:
		\newsavebox{\GtwooverP}
		\savebox{\GtwooverP}{$\faktor{G_2}{P}$}
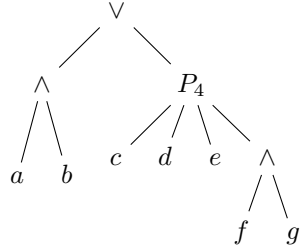
		\begin{center}	
			
			\begin{forest}
				[ $\vee$ [ $\wedge$  [$a$, l = 12.2mm] [$b$, l = 12.2mm]   ]   
				[ $\pfour$  [$c$] [$d$] [$e$]  [$\wedge$ [$f$] [$g$]] ]] 
			\end{forest} 
			
		\end{center}
		
	\end{ex}	
	
	
%
		\subsection{Maximal cliques and stable sets via modular decomposition}
%

In this subsection we outline an algorithm for evaluation that operates recursively on the modular decomposition tree. 
This essentially reduces the problem of evaluation to the prime graphs.

			\begin{observation}
				Let $G= (V,E)$ be a graph and $P$ a modular partition.
			The map $\phi : V \to V/P$ by $\phi (x) = M$ unique such that $M\ni x$ induces a homomorphism from $G$ to $G/P$.
			Therefore the images of cliques under $\phi$ are again cliques.
		\end{observation}
		
		In fact the particular homomorphisms induced by quotients preserve a lot more structure. For instance stable sets are also preserved, as well as maximality, as we will now show.
		The point of this is that,
		in order to evaluate graphs recursively on their modular decomposition, we need to first classify their maximal cliques and stable sets in this way.

		\begin{lemma}
						\label{lem:mc-ms-ser-par-comp}
			Let $G=(V,E)$ be a graph and $P\subseteq\pow(V)$ be a modular partition. 
			For $X\subseteq V$ we have:
			\begin{itemize}
				\item $X \in \mc(G)$ iff 
				there exists $S \in \mc(G/P)$ and some $S_M \in \mc(M)$, for $M\in S$, s.t.\ $X = \bigcup\limits_{M \in S} S_M$.
				\item $X \in \ms(G)$ iff there exists $T \in \ms(G/P)$ and some $T_M \in \mc(M)$, for $M\in S$, s.t.\ $X = \bigcup\limits_{M \in T} T_M$.
			\end{itemize}
		\end{lemma}
			
		\begin{proof}
			[Proof sketch]
			We prove only the statements regarding maximal cliques, the ones for maximal stable sets following by duality.
			
			For the left-right implication, let $X \in \mc(G)$ and set
			$S = \{M \in P : X\cap M \neq \emptyset \}$.
			\begin{itemize}
				\item $S$ is a clique of $G/P$: for any distinct $M,M' \in S$ there are some $x\in M$ and $y \in M'$ by nonemptiness, and we have $\redge x y $ in $G$ since $x,y \in X \in \mc(G)$.
				Thus $\redge M {M'}$ in $G/P$ by modularity.
				\item $S$ is maximal: suppose there is some $M'\notin S$ such that $\redge {M'} S$ in $G/P$. Any $x\in M'$ is not in $ X$, by definition of $S$, and we have $\redge x M $ in $G$ for all $M \in S$. Therefore $\redge x X$ in $G$ (since $P$ partitions $V$) contradicting maximality of $X$.
			\end{itemize}
		Now we may define $S_M = X \cap M$ for each $M\in S$.
			\begin{itemize}
				\item Each $S_M$ is a clique of $M$, since $S_M\subseteq X\in \mc(G)$.
				\item Each $S_M $ is maximal: suppose there is some $x \in M\setminus S_M$ with $\redge x {S_M} $ in $M$. 
				By modularity we must have that $\redge x  S$ in $G$ and hence already $x \in S$ by maximality of $S$.
			\end{itemize}
		Since $P$ is a partition of $V$, we also have that $\{S_M : M\in S \}$ partitions $X$, so $X = \bigcup\limits_{M\in S} S_M$, as required.
			
			For the right-left implication, suppose $S\in \mc(G/P)$ and $S_M\in \mc(M)$, for $M\in S$, s.t.\ $X = \bigcup\limits_{M\in S} S_M$.
			To show that $X$ is a clique, let $x,y \in X$ be distinct. We have two cases:
			\begin{itemize}
				\item there is $M$ s.t.\ $x,y \in M$, in which case $x,y \in S_M$ so $\redge x y $ in $G$; or,
				\item there are disjoint $M,M'$ with $x\in M$ and $y\in M'$, in which case $\redge M {M'} $ in $S$ and so also $\redge x y $ in $G$. 
			\end{itemize}
		For maximality, suppose $x \notin X$ s.t.\ $\redge x X$ in $G$. Again we have two cases:
		\begin{itemize}
			\item if there is $M\in S$ with $x\in M$, then $S_M$ can be extended by $x$;
			\item if $x \in M \notin S$, then we must have $\redge M S$ by modularity, and so $S$ can be extended by $M$.
		\end{itemize}
	In either case we have a contradiction, concluding the proof.
		\end{proof}

		Since the disjunctive and conjunctive nodes in the decomposition tree are special cases of a modular partition, the following characterisation is now immediate from the preceding lemma:
		\begin{proposition}
			[Maximal cliques and stable sets via modular decomposition]
			\label{prop:mc-ms-via-decomposition}
Let $G$ be a graph and $X\subseteq V(G)$. We have:
			\begin{enumerate}
				\item If $G$ is the singleton graph $\{x\}$ then $\mc(G) = \ms(G) =  \{\{x\}\}$.
				\item 
				If $\moddec (G) = \dis (\tree(G_1), \dots, \tree(G_n) )$
then:
\begin{itemize}
	\item $X \in \mc(G)$ if and only if, for some $i$, $X \in \mc(G_i)$.
	\item $X \in \ms(G)$ if and only if, for every $i$, $X\cap V(G_i) \in \ms(G_i)$.
\end{itemize} 
				\item 
If $\moddec (G) = \con (\tree(G_1), \dots, \tree(G_n))$
				then:
				\begin{itemize}
					\item $X \in \mc (G)$ if and only if, for every $i$, $X\cap V(G_i) \in \mc(G_i)$.
					\item $X \in \ms(G)$ if and only if, for some $i$, $X \in \ms(G_i)$.
				\end{itemize}
				\item 
If $\moddec(G) = (G/P)(\tree(M_1), \dots , \tree(M_n))$
%
				then,
				\begin{itemize}
					\item $X \in \mc(G)$ iff 
					there exists $S \in \mc(G/P)$ and some $S_i \in \mc(M_i)$, for $M_i \in S$, s.t.\ $X = \bigcup\limits_{M_i \in S} S_i$.
					\item $X \in \ms(G)$ iff there exists $T \in \ms(G/P)$ and some $T_i \in \ms (M_i)$, for $M_i \in T$, s.t.\ $X = \bigcup\limits_{M_i \in T} T_i$ for some $T_i \in \mc(M_i)$.
				\end{itemize}
			\end{enumerate}
		\end{proposition}

%
		
		\subsection{Evaluation via modular decomposition}
	We give a characterisation of evaluation by recursion on the decomposition tree of graphs.
	In effect, this yields an algorithm for evaluation by reduction to prime graphs.
	
	Since our semantics is multivalued, we have to be a little careful with how we construct assignments during recursion on the decomposition tree.
	\begin{definition}
		[Positive and negative quotient assignments]
		Let $G=(V,E)$ be a graph and $P\subseteq \pow(V)$ be a modular partition.
		For an assignment $X\subseteq V$, we define the following subsets of $P$, relative to $G$:
		\[
		\begin{array}{rcl}
		\posquot_P ( X) & := & \{ M\in P : M(X,1) \} \\
			\negquot_P (X) & := & \{ M\in P : \neg M(X,0) \}
		\end{array}
		\]
		
	\end{definition}
	
	\todo{CHECK THE NEGATIVE PART!}
	
	\begin{lemma}
		Let $G=(V,E)$ be a graph and $P\subseteq \pow(V)$ be a modular partition.
		For an assignment $X\subseteq V$ we have:
		\begin{enumerate}
			\item $G(X,1)$ if and only if $(G/P) (\posquot_P (X),1 )$
			\item $G(X,0)$ if and only if $(G/P) (\negquot_P(X),0)$.
		\end{enumerate}
	\end{lemma}
\begin{proof}
	We prove only (1), the case of (2) following by duality.
	
	For the left-right implication, let $S\in \mc(G)$ with $S\subseteq X$.
	By Lemma~\ref{lem:mc-ms-ser-par-comp}, there is $S_P \in \mc(G/P)$ s.t.\ $S = \bigcup\limits_{M \in S_P} S_M$ for some $S_M \in \mc (M)$, for $M \in S_P$.
	Now, for each $M\in S_P$, we have that $S_M \subseteq S \subseteq X$, so $M(X,1)$ and $S_M \in \posquot_P (X)$. 
	Thus we have that $(G/P) (\posquot_P (X),1)$, as required.
	
	For the right-left implication, let $S \in \mc(G/P)$ with $S\subseteq \posquot_P(X)$.
	By definition of $\posquot_P(X)$ we have $\forall M \in S. M(X,1)$.
	so for each $M\in S$ fix $S_M \in \mc(M)$ s.t.\ $S_M \subseteq X$.
	Now, by Lemma~\ref{lem:mc-ms-ser-par-comp} we have $S'\in \mc(G)$ with $S' = \bigcup\limits_{M\in S} S_M$.
	Since $S_M\subseteq X$ for each $M\in S$, we also have $S'\subseteq X$, so $G(X,1) $ as required.	
	\end{proof}
	
	The following result, computing evaluation by recursion on a decomposition tree, is now immediate from the Lemma above and Proposition~\ref{prop:mc-ms-via-decomposition}:
	\begin{theorem}
		[Evaluation by recursion on decomposition trees]
		Let $G$ be a graph and $X \subseteq V(G)$.
		We have:
		\begin{enumerate}
			\item If $G$ is a singleton graph with $V(G) = \{x\}$ then: 
			\begin{itemize}
				\item $G(X,1)$ if and only if $x \in X$.
				\item $G(X,0)$ if and only if $x \notin X$.
			\end{itemize}
		\item If $\tree(G) = \dis(\tree(G_1), \dots , \tree(G_n))$ then:
		\begin{itemize}
			\item $G(X,1)$ if and only if, for some $i$, $G_i(X,1)$.
			\item $G(X,0)$ if and only if, for every $i$, $G_i (X,0)$.
		\end{itemize}
		\item If $\tree(G) = \con(\tree(G_1), \dots , \tree(G_n))$ then:
	\begin{itemize}
		\item $G(X,1)$ if and only if, for every $i$, $G_i(X,1)$.
		\item $G(X,0)$ if and only if, for some $i$, $G_i (X,0)$.
	\end{itemize}
\item If $\tree(G) = (G/P)(\tree(G_1), \dots, \tree(G_n))$ then:
\begin{itemize}
	\item $G(X,1)$ if and only if $(G/P) (\posquot_P (X),1 )$
	\item $G(X,0)$ if and only if $(G/P) (\negquot_P(X),0)$.
\end{itemize}
		\end{enumerate}
%
	\end{theorem}
		
%
%
%
		
%
%
%
%
%

\section{Positional games beyond cographs}
\label{sect:games}
A refined view of evaluation can be given in a game theoretic setting;
in particular, \quot{extensive game forms} forms have been studied in connection with read-once formulas, leading to the notion of \emph{positional games}, e.g.\ \cite{gurvich1982normal,golumbic_gurvich_crama_hammer_2011}.
In this section, we will see how to extend the `evaluation game' on formulas to arbitrary graphs. Along the way, we recover a distinction between the static and sequential versions of this game, a distinction which does not exist at the level of formulas.

The evaluation game of a formula is played as follows. Two players construct a path through the formula tree, with one player  (`Eloise') choosing directions at disjunction nodes and the other (`Abelard') choosing directions at conjunction nodes. The possible `outcomes' of the play are then the leaves of the tree, i.e.\ the variables of the formula, and a notion of winning can be imposed via some Boolean payoff set (i.e.\ an assignment). 
In this way, the strategies of the Eloise and Abelard are determined precisely by the minterms and maxterms, respectively, of the formula they are playing on.
One interesting pursuit is to establish which game forms correspond to evaluation games, a question resolved for formulas in \cite{gurvich1982normal}, though extending this result to arbitrary graphs is beyond the scope of this work.

Let us consider an example of the evaluation game on formulas.
\begin{ex}
		[Evaluation games]
	\label{Ex:GameGurvich}
Let $A$ be the following formula:
\[
w_0 \dis ( x \con (y \dis (z_0 \con z_1 \con z_2 \con z_3 )) ) \dis w_1
\]
The formula tree $\tree(A)$ is as follows:\footnote{Notice that there is some mismatch here in the sense that formulas have binary connectives but our decomposition trees allow arbitrary fan-in. We gloss over this mismatch in this example, and wherever it is not ambiguous.}
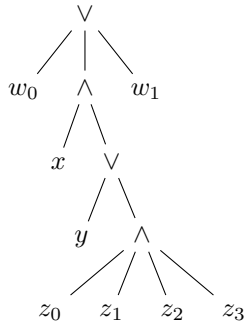
	\begin{center}	
		\begin{forest}
			[ $\vee$ 
			[$w_0$]  [$\wedge$ [$x$] [$ \vee $ [$y$] [$ \wedge$ [$z_0$] [$z_1$] [$z_2$] [$z_3$] ] ]
			] [$w_1$] ]
		\end{forest} 
	\end{center}
	A possible play of the corresponding game would be as follows:
	\begin{itemize}
		\item The root of the tree is labelled with a $ \vee $, so Eloise plays one of its children.
		If she chooses $w_0$ or $w_1$ then the play ends with outcome $w_0$ or $w_1$, respectively. Suppose she picks the $\con$ child.
		\item It is now Abelard to choose, and he has two choices: either $x$, again terminating the play with outcome $x$, or $\vee$. Suppose he chooses the latter.
		\item At this point let us suppose that Eloise ends the play by choosing $y$.
	\end{itemize} 

The strategy employed by Eloise in this play corresponds to the minterm $\{x,y \}$, in the sense that she always makes a choice that admits the possibility of an outcome in that set. 
Abelard's strategy corresponds to several maxterms, namely $\{w_0,w_1,y, z_i  \}$ for any $i<4$, for the same reason. The play terminated before Abelard had to choose a $z_i$, which is why it is consistent with multiple strategies.

Equipped with an assignment $X \subseteq\{w_0,w_1,x,y,z_0, z_1,z_2,z_3 \}$, we may say that the play above is \emph{winning} for Eloise if the outcome $y$ is in $X$, and otherwise winning for Abelard.
\end{ex}


\subsection{Games on arbitrary graphs}

In order to extend the evaluation game to arbitrary graphs, it thus suffices to establish rules of play at prime nodes in the decomposition tree.
As hinted at in the previous example, one way to see the play is that Eloise is playing according to some maximal clique whereas Abelard plays according some maximal stable set.
We directly import this idea into the following definition of games on arbitrary graphs, though we will see that the games induced become more sensitive to the format of gameplay than before.
%
%
%
%
%

\begin{definition}
	[Games on graphs]
	\label{dfn:games-on-graphs}
	We define a two-player game, with players Eloise and Abelard by recursion on the decomposition tree of a graph. \emph{Plays} of the game have \emph{outcome} that is either a variable $x\in \Var$ or $\emptyset$.
	
	Let $G$ be a graph, and let us define a {play} of $G$ and the {outcome} of a play as follows: 
	\begin{enumerate}
		\item If $V(G)$ is the singleton $\{x\}$ then the play is empty and the outcome is $x$.
		\item if $\moddec(G)$ has the form $\vee (\tree(G_1), \dots, \tree(G_n))$,
		then Eloise chooses one of the $G_i$s and the play continues on $G_i$.
		
\item If $\moddec(G)$ has the form
$\wedge (G_1, \dots, G_n)$,
then Abelard chooses one of the $G_i$s and the play continues on $G_i$.

\item If $\moddec(G)$ has the form 
$(G/P) (M_1, \dots, M_n)$,
then: 
\begin{itemize}
	\item Eloise chooses some $S \in \mc(G/P)$; and,
	\item Abelard chooses some $T \in \ms(G/P)$,
\end{itemize}
and the play continues on the unique module $M_i$ in the intersection $S\cap T$, if it exists. 
If there is a deadlock, i.e.\ $S\cap T = \emptyset$, then the play ends with outcome $\emptyset$.
	\end{enumerate}
\end{definition}

We construe plays as the sequence of graphs induced by the choices of Eloise and Abelard following the rules above. This sequence always determines a path through the decomposition tree of $G$ from the root to either a leaf, if the outcome is the corresponding variable, or an internal prime node, if the play hits a deadlock.

\begin{remark}
	[Deadlocks and determinism]
	Note that on deterministic graphs maximal cliques and stable sets always intersect, since they are CIS, and therefore a play always has a non-empty outcome.
\end{remark}

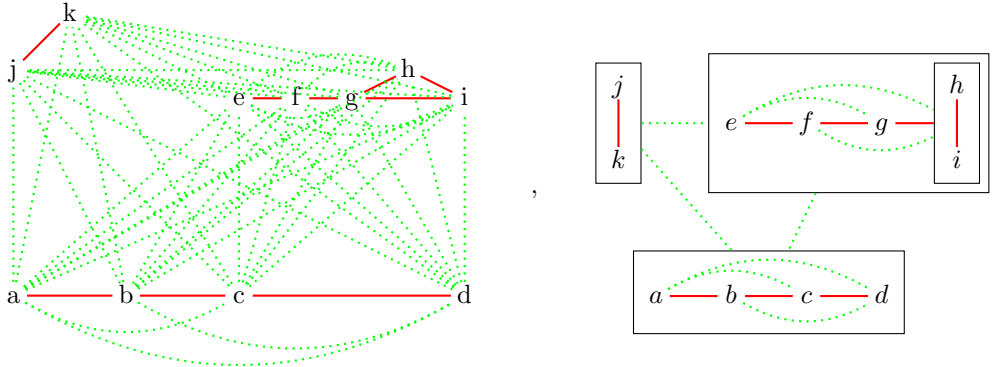
\begin{figure}
	\label{fig:big-graph-and-modular-presentation}
	\[
	\raisebox{-0.5\height}{\begin{tikzpicture}[scale=1.5]
		\node[vertex] (v1) at (0,0) {a};
		\node[vertex] (v2) at (1,0) {b};
		\node[vertex] (v3) at (2,0) {c};
		\node[vertex] (v4) at (4,0) {d};
		\node[vertex] (v5) at (2,1.75) {e};
		\node[vertex] (v6) at (2.5,1.75) {f};
		\node[vertex] (v7) at (3,1.75) {g};
		\node[vertex] (v8) at (3.5,2) {h};
		\node[vertex] (v9) at (4,1.75) {i};
		\node[vertex] (v10) at (0,2) {j};
		\node[vertex] (v11) at (0.5,2.5) {k};

		\draw[r] (v1) edge (v2) ;
		\draw[lg, bend right] (v1) edge (v3) ; 
		\draw[lg, bend right] (v1) edge (v4) ; 
		\draw[lg] (v1) edge (v5) ;
		\draw[lg] (v1) edge (v6) ;
		\draw[lg] (v1) edge (v7) ;
		\draw[lg] (v1) edge (v8) ;
		\draw[lg] (v1) edge (v9) ;
		\draw[lg] (v1) edge (v10);
		\draw[lg] (v1) edge (v11) ;
		
		\draw[r] (v2) edge (v3) ;
		\draw[lg, bend right] (v2) edge (v4) ;		
		\draw[lg] (v2) edge (v5) ;
		\draw[lg] (v2) edge (v6) ;
		\draw[lg] (v2) edge (v7) ;
		\draw[lg] (v2) edge (v8) ;
		\draw[lg] (v2) edge (v9) ;
		\draw[lg] (v2) edge (v10);
		\draw[lg] (v2) edge (v11) ;

		\draw[r] (v3) edge (v4) ;
		\draw[lg] (v3) edge (v5) ;
		\draw[lg] (v3) edge (v6) ;
		\draw[lg] (v3) edge (v7) ;
		\draw[lg] (v3) edge (v8) ;
		\draw[lg] (v3) edge (v9) ;
		\draw[lg] (v3) edge (v10);
		\draw[lg] (v3) edge (v11) ;
		
		\draw[lg] (v4) edge (v5) ;
		\draw[lg] (v4) edge (v6) ;
		\draw[lg] (v4) edge (v7) ;
		\draw[lg] (v4) edge (v8) ;
		\draw[lg] (v4) edge (v9) ;
		\draw[lg] (v4) edge (v10);
		\draw[lg] (v4) edge (v11) ;
		
		\draw[r] (v5) edge (v6) ;
		\draw[lg, bend right] (v5) edge (v7) ;
		\draw[lg, bend left] (v5) edge (v8);
		\draw[lg, bend right] (v5) edge (v9) ; 
		\draw[lg] (v5) edge (v10);		
		\draw[lg] (v5) edge (v11);
		
		\draw[r] (v6) edge (v7) ;	
		\draw[lg, bend left] (v6) edge (v8);
		\draw[lg, bend right] (v6) edge (v9);
		\draw[lg] (v6) edge (v10);	
		\draw[lg] (v6) edge (v11);
		
		\draw[r] (v7) edge (v8) ;
		\draw[r] (v7) edge (v9) ;
		\draw[lg] (v7) edge (v10);
		\draw[lg] (v7) edge (v11);
		
		\draw[r] (v8) edge (v9) ;
		\draw[lg] (v8) edge (v10);	
		\draw[lg] (v8) edge (v11);
		
		\draw[lg] (v9) edge (v10);
		\draw[lg] (v9) edge (v11);
		
		\draw[r] (v10) edge (v11) ;
		\end{tikzpicture}}
	\qquad , \qquad
	\raisebox{-0.5\height}{\begin{tikzpicture}
		\node[vertex] (j) at (1.5,0.5) {$j$};
		\node[vertex] (k) at (1.5,-0.5) {$k$};
		
		\draw[r] (j) -- (k);
		
		\node[draw, fit=(j) (k) ] (M1) {} ;
		
		\node[vertex] (e) at (3,0) {$e$};
		\node[vertex] (f) at (4,0) {$f$};
		\node[vertex] (g) at (5,0) {$g$};
		
		\node[vertex] (h) at (6,0.5) {$h$};
		\node[vertex] (i) at (6,-0.5) {$i$};
		\draw[r] (h) -- (i);
		
		\node[draw,fit=(h) (i) ] (M3) {};
		
		\draw[r] (e) -- (f);
		\draw[r] (f) -- (g);
		\draw[r] (g) -- (M3);
		\draw[g,bend left] (e) edge (g);
		\draw[g,bend left] (e) edge (M3);
		\draw[g,bend right] (f) edge (M3);
		
		\node[draw, fit=(e) (f) (g) (M3)] (M2) {};
		
		\node[vertex] (a) at (2,-2.3) {$a$};
		\node[vertex] (b) at (3,-2.3) {$b$};
		\node[vertex] (c) at (4,-2.3) {$c$};
		\node[vertex] (d) at (5,-2.3) {$d$};
		
		\draw[r] (a) -- (b);
		\draw[r] (b) -- (c);
		\draw[r] (c) -- (d);
		\draw[g,bend left] (a) edge (c);
		\draw[g,bend left] (a) edge (d);
		\draw[g,bend right] (b) edge (d);
		\node[vertex] (vpoint) at (3.5,-2) {};
		\node[vertex] (upoint) at (3.5,-2.5) {};
		
		\node[draw,fit=(a) (b) (c) (d) (upoint) (vpoint) ] (M4) {};

		\draw[g] (M1) -- (M2);
		\draw[g] (M1) -- (M4);
		\draw[g] (M2) -- (M4);
		\end{tikzpicture}}
	\]
	\caption{An example of a graph and its modular presentation, obtained by recursively partitioning the graph into maximal proper modules.}
\end{figure}

\begin{ex}
	[A play of a graph]
	\label{Ex:GraphGame}
	Let $ G $ be the graph in Figure~\ref{fig:big-graph-and-modular-presentation}, given in both the usual red-green presentation, left, and its `modular' presentation, right.
The decomposition tree $\tree(G)$ is:
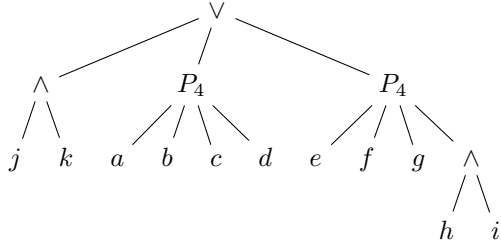
\[
\begin{forest}
[
$\vee$ 
[
$\wedge$ [$j$] [$k$]
]
[
$\pfour$ [$a$] [$b$] [$c$] [$d$]
]
[
$\pfour$ [$e$] [$f$] [$g$] [
$\wedge$ [$h$] [$i$]
]
]
]
\end{forest}
\]
	We will now run through a possible play on $G$:
	\begin{enumerate}
		\item The root of $\tree(G)$ is $\dis$, so Eloise plays first and chooses the rightmost child, $\pfour$, corresponding to the following maximal proper module, left, and subtree, right:
			\[
%
%
%
%
			\raisebox{-0.5\height}
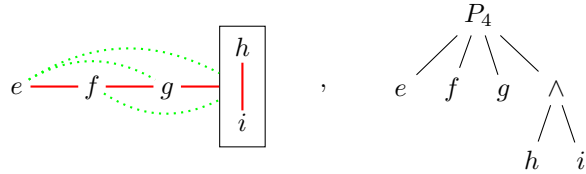
{\begin{tikzpicture}
			\node[vertex] (e) at (3,0) {$e$};
			\node[vertex] (f) at (4,0) {$f$};
			\node[vertex] (g) at (5,0) {$g$};
			
			\node[vertex] (h) at (6,0.5) {$h$};
			\node[vertex] (i) at (6,-0.5) {$i$};
			\draw[r] (h) -- (i);
			
			\node[draw,fit=(h) (i) ] (M3) {};
			
			\draw[r] (e) -- (f);
			\draw[r] (f) -- (g);
			\draw[r] (g) -- (M3);
			\draw[g,bend left] (e) edge (g);
			\draw[g,bend left] (e) edge (M3);
			\draw[g,bend right] (f) edge (M3);
			\end{tikzpicture}}
			\qquad,\qquad
			\raisebox{-0.5\height}{\begin{forest}
				[
				$\pfour$ [$e$] [$f$] [$g$] [
				$\wedge$ [$h$] [$i$]
				]
				]
				\end{forest}}
			\]	
		\item The root is now $\pfour$, which is prime. 
		\begin{itemize}
			\item Eloise chooses the maximal clique $\{g,\{h,i\} \}$, consisting of the third and fourth child.
			\item Abelard chooses the maximal stable set $\{e, \{h,i\} \}$, consisting of the first and fourth child.
		\end{itemize}
	\item The intersection is nonempty, corresponding to the following module, left, and subtree, right:
	\[
\raisebox{-0.5\height}
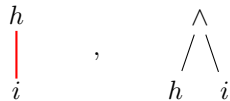
{
\begin{tikzpicture}
	\node[vertex] (h) at (6,0.5) {$h$};
\node[vertex] (i) at (6,-0.5) {$i$};
\draw[r] (h) -- (i);
\end{tikzpicture}
}
	\qquad,\qquad
	\raisebox{-0.5\height}{\begin{forest}
		[
		$\wedge$ [$h$] [$i$]
		]
		\end{forest}}
	\]
	Since the root is $\con$, it is Abelard to play, and he chooses $h$.
	\item  $h$ is a leaf of the decomposition tree, so the play ends with outcome $h$.
	\end{enumerate}
If at step (2) Eloise had chosen $\{f,g \}$ and Abelard had chosen $\{e,\{h,i\} \}$ then the play would end with outcome $\emptyset$, since the intersection of these sets is empty.
\end{ex}



\subsection{Boolean payoffs and strategies}

In order to talk about the `winner' of a game, we introduce standard Boolean payoff sets, playing the role of assignments earlier.

\begin{defi}[Winning under an assignment]
	Let $G= (V,E)$ be a graph and $X\subseteq V$ be an assignment. 
	Given a play of $G$, we say that:
\begin{itemize}
	\item Eloise wins if the outcome of the play is some $x\in X$.
	\item Abelard wins if the outcome of the play is some $x \notin X$.
	\item The play is a draw if the outcome of the play is $\emptyset$.
\end{itemize}
\end{defi}
%

As mentioned earlier, evaluation games with Boolean payoff gave us a game-theoretic characterisation of evaluation in the case of formulas, equivalently cographs. There the situation was simple: evaluating to $1$ corresponds to a winning strategy for Eloise, and evaluating to $0$ corresponds to a winning strategy for Abelard.
Here it is not so simple, not only because our semantics is relational but also since determinacy is sensitive to the `mode' of play.

Recalling Definition~\ref{dfn:games-on-graphs}, notice that we did not specify at a prime node whether Eloise and Abelard make their choices independently of each other, or whether one player is able to make their choice once the other has already declared theirs.
In the formula setting this makes no difference: the player with a winning strategy has a uniform strategy no matter the moves of the opponent, - just play according to some minterm or maxterm.
For arbitrary graphs it turns out that there is an advantage for the second player, since they may react to the choices of the first.

\begin{defi}[Sequential and static strategies]
	Let $G=(V,E)$ be a graph.
	A \emph{strategy} on $G$ is a specification of choices for a player at each relevant node of $\tree(G)$.\footnote{Note that we only consider `positional' strategies here, where the choice of a player does not depend on the history of the play. It is not hard to see that this is sufficiently general for determinacy.} 
	Formally, a strategy for Eloise is a map that associates:
	\begin{itemize}
		\item to each $\dis$ node of $\tree(G)$ a child of that node; and,
		\item to each prime node $G/P$ of $\tree(G)$ some $S \in \mc(G/P)$.
	\end{itemize}
Dually, a strategy for Abelard is a map that associates:
\begin{itemize}
	\item to each $\con$ node of $\tree(G)$ a child of that node; and,
	\item to each prime node $G/P$ of $\tree(G)$ some $T \in \ms(G/P)$.
\end{itemize}

Furthermore, we may distinguish different \emph{modes} of strategy:
\begin{itemize}
	\item If a strategy $\sigma$ for a player $p$ at prime nodes may depend on the other player's choice, then we call $\sigma$ a \emph{reactionary} strategy for $p$.
	\item Otherwise, if a strategy $\sigma$ for a player $p$ does not depend on the choice of the other player at prime nodes, then we call $\sigma$ a \emph{static} strategy for $p$.
\end{itemize}
	

We say that a $p$-strategy $\sigma$ is \emph{winning} with respect to an assignment $X$ if every play according to $\sigma$ results in a win for $p$.
We say that $\sigma$ is \emph{drawing} with respect to $X$ if every play according to $\sigma$ results in a win or draw for $p$.
\end{defi}

%
%



\subsection{Characterisation of evaluation via static strategies}
Let us first consider the somewhat simpler case of static games, since the results and arguments are similar to the formula setting, only accounting for nonfunctionality and the possibility of draws.
Table~\ref{table:offline-games} summarises the circumstances in which each of the players have a winning or drawing strategy, given an assignment $ X $. 
This of course depends on what $G$ may evaluate to on $X$, but we also distinguish whether $G$ is deterministic or not, since in the former case draws are impossible, by the CIS property.
\begin{table}[t]
	\begin{small}
		\begin{center}
			\textbf{Static games:}
			\begin{tabularx}{\textwidth}{l|>{\centering\arraybackslash}m{.95in}|>{\centering\arraybackslash}m{.95in}|>{\centering\arraybackslash}m{.95in}|>{\centering\arraybackslash}m{1in}}
				\hlineB{4}
				&$ G(X) = \{1\} $ & $ G(X)=\{0\} $ & $G(x)=\emptyset$ & $G(X)=\{0,1\}  $\\ \hline
				
				Det. & 
				Eloise has a winning strategy & 
				Abelard has a winning strategy & 
				- & 
				N/A \\ \hline
				
				Non-det. & 
				Eloise has a drawing strategy & 
				Abelard has a drawing strategy &
				- &  
				Eloise and Abelard have a drawing strategy\\ 	 \hlineB{4}
			\end{tabularx}
			\vspace{1mm}
			\caption{Winning/drawing strategies for static games. We write $G(X) = B$ if $G(X,b) \iff b \in B$.	 \quot{-} denotes the situation where neither player has a winning or drawing strategy, and N/A denotes a situation that cannot occur.
		}
		\end{center}
	\end{small}
\label{table:offline-games}
\end{table}				
The entries of the table are justified by the following result:
\begin{thm}
	\label{thm:static-determinacy}
	Let $ G  $ be a graph and $ X \subseteq V(G) $ an assignment. We have the following:
	\begin{enumerate}
		\item\label{item:eloise-drawing-static} $G(X,1)$ if and only if Eloise has a static drawing strategy.
		\item\label{item:abelard-drawing-static} $G(X,0)$ if and only if Abelard has a static drawing strategy.
%
	\end{enumerate}
If $G$ is deterministic then the above can be strengthened to static winning strategies.
\end{thm}

\begin{proof}
Notice that the final clause on deterministic graphs follows immediately from (1) and (2) by the CIS property, Proposition~\ref{prop:det-iff-cis}.
	
For the left-right implication of \ref{item:eloise-drawing-static}, the idea is that 
Eloise plays to remain in a maximal clique inside the given assignment.
Formally, given a graph $G$, an assignment $X$, and some maximal clique $S\subseteq X$, we describe what it means for Eloise to play \emph{according to} $S$, appealing to the classification of maximal cliques in Propostion~\ref{prop:mc-ms-via-decomposition}:
\begin{itemize}
	\item If $\tree(G)$ is a singleton then there is nothing to play and, since $S$ is nonempty, Eloise wins. 
	\item If $\tree(G) = \dis (\tree (G_1), \dots, \tree(G_n)$, then we must have $S \in \mc(G_i)$, and so Eloise chooses this $G_i$ and continues playing according to $S$.
	\item If $\tree(G) = \con(\tree (G_1), \dots, \tree (G_n))$, then we must have that $S_i = S \cap V(G_i) \in \mc(G_i)$, so if Abelard chooses $G_i$ Eloise continues playing according to $S_i$.
	\item If $\tree(G ) = (G/P) (\tree (M_1), \dots, \tree (M_n) )$, then $S = \bigcup\limits_{i \in I} S_i$, for some $\{M_i \}_{i \in I} \in \mc(G/P)$ and $S_i \in \mc (M_i)$. Eloise plays the maximal clique $\{M_i\}_{i \in I}$ and, if the play continues on some module $M_i$, with $i\in I$, then she continues playing according to $S_i$.
\end{itemize}
 There are two options for a play of this strategy: either a variable is reached (the first case) and Eloise wins, or at a prime node (fourth case) the intersection of Eloise's and Abelard's sets is empty and the game is drawn.
%

Conversely, suppose $\sigma$ is a static drawing strategy for Eloise and write $\tree(\sigma)$ for the subtree of $\tree(G)$ induced by it. Formally, $\tree(\sigma)$ is the smallest subtree satisfying:
\begin{itemize}
	\item The root of $\tree(G)$ belongs to $\tree(\sigma)$;
	\item For any $\dis$ node in $\tree(\sigma)$ the child chosen by $\sigma$ belongs to $\tree(\sigma)$;
		\item For any $\con $ node in $\tree(\sigma)$, all its children also belong to $\tree(\sigma)$;
		\item For a prime node $G/P$ in $\tree(\sigma)$, all children in the maximal clique chosen by $\sigma$ belong to $\tree(\sigma)$.
\end{itemize}
Let us write $L(\sigma)$ for the set of leaves of $\tree(\sigma)$. We have the following:
\begin{itemize}
	\item $L(\sigma)$ is a clique of $G$: let $x,y\in L(\sigma)$ be distinct and let us inspect the root of the smallest subtree of $\tree(\sigma)$ (equivalently $\tree(G)$) containing both $x$ and $y$. It cannot be a $\dis$ node by construction of $\tree(\sigma)$, and if it is a $\con$ or $G/P$ node, then $\redge x y $ in $G$ by construction of $\tree(\sigma)$.
	\item $L(\sigma)$ is maximal: suppose not and let $z \notin L(\sigma)$ such that $\redge{z}{L(\sigma)}$ in $G$. Let us take the smallest subtree of $\tree(\sigma)$ whose root $\nu$ induces a subtree of $\tree(G)$ containing $z$.
	$\nu$ cannot be a $\dis $ node, since that contradicts $\redge z {L(\sigma)}$; $\nu $ cannot be a $\con $ node since that contradicts leastness of $\nu$; finally, if $\nu$ is a prime node $G/P$, then by modularity we have that the clique chosen by $\sigma$ can be extended to the module containing $z$, contradicting maximality.
\end{itemize}
%
\ref{item:abelard-drawing-static} follows by duality.
%
%
%
%
%
%
\end{proof}

If $ X  $ evaluates to neither 0 nor 1, it follows that neither player necessarily has a drawing/winning strategy in a static setting. 
Note that this does not contradict finite determinacy, which a priori holds only for sequential games.
\begin{example}
	[Non-determinacy]
	\label{ex:non-determinacy-static}
Let us revisit the $ P_4 $:
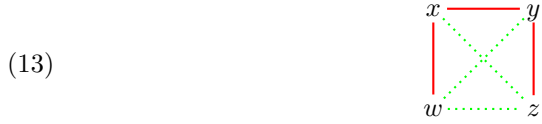
\begin{equation}
\label{eqn:pfour-non-determinacy-static}
\FourGraph{x,y,w,z}rrggrg
\end{equation}
Recall that under the assignment $X = \{w,z \}$, we have neither $G(X,0)$ nor $G(X,1)$.
It turns out that neither player has a static drawing strategy in this circumstance, as we now argue.
Since the $\pfour$ is prime, the corresponding game is one-shot, so it suffices to show that neither player has a drawing move.
\begin{itemize}
	\item If Eloise plays $\{x,y\}$ then she will lose if Abelard plays $\{x,z\} $ or $\{w,z\}$;
	\item If Eloise plays $\{x,w\}$ then she will lose if Abelard plays $\{x,z\}$;
	\item If Eloise plays $\{y,z \}$ then she will lose if Abelard plays $\{w,y\}$.
\end{itemize}
Since the $\pfour$ is isomrophic to its dual, it is immediate that Abelard too has no drawing strategy.
%
\end{example}

As we will see in the next subsection, finite derminacy \emph{does} apply to sequential games, and in situations like the $\pfour$ above, it is the second player who actually has a winning strategy.

It is natural to wonder whether any of the results of Table~\ref{table:offline-games} can be strengthened from drawing to winning.
This is not the case, as we can see from the following example.
\begin{example}
	\label{ex:5-cycle-no-forced-win-static}
Let us recall the $5$-cycle:
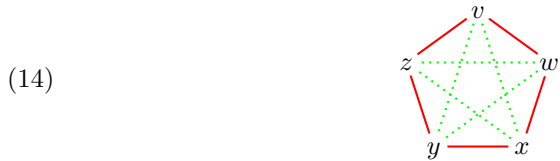
\begin{equation}
\label{eqn:5-cycle-no-forced-win-static}
			\FiveGraphCirc{v,w,x,y,z}rggrrggrgr
\end{equation}
Recall that this graph is not deterministic (though it is total) since the assignment, say, $\{v,w\}$ evaluates to both $0$ (being disjoint from the maximal stable set $\{x,z\}$) and $1$ (being a maximal clique itself).
It is also a prime graph and so the corresponding game is one-shot.

Now, if we take the slightly larger assignment $\{v,z,w\}$, the graph still evaluates to $1$ but no longer evaluates to $0$.
Nonetheless, Eloise cannot force a win: whichever maximal clique she chooses, Abelard may choose a disjoint maximal stable set, as reasoned above, by the symmetry of the graph.
\end{example}

\subsection{Characterisation of evaluation via sequential strategies}

For sequential games, the second player has an advantage since they have strictly more information to exploit at each prime node, namely the opponent's move at that node.

\begin{ex}
	[Winning sequentially]
	Revisiting Example~\ref{ex:5-cycle-no-forced-win-static}, where the drawing player could not force a win, let us consider what happens in a sequential setting.
	Taking the same graph \eqref{eqn:5-cycle-no-forced-win-static} and the same assignment $\{v,z,w\}$, it is clear that Eloise can force a win since every maximal stable set intersects the assignment, and so she can react to Abelard's move with the appropriate maximal clique.
%
\end{ex}

What is more, finite determinacy now applies and so the game is completely determined.
\begin{example}
	[Determinacy and non-totality]
	Revisiting Example~\ref{ex:non-determinacy-static}, where neither player could force even a draw, let us consider what happens in a sequential setting. 
	Taking the same graph \eqref{eqn:pfour-non-determinacy-static} and assignment $X= \{w,z\}$, it turns out that the second play can actually force a win.
In fact, the case analysis of Example~\ref{ex:non-determinacy-static} immediately shows that Abelard wins when playing second.
\end{example}

Table~\ref{table:online} summarises the circumstances in which a player has a winning or drawing reactionary strategy. The results for static strategies naturally still hold for the first player. 
Again, we distinguish deterministic graphs from general graphs.

\begin{table}[t]			
	\begin{small}
		\begin{center}		
			\vspace{3mm}
			\textbf{Sequential games:}			
			\begin{tabularx}{\textwidth}{l|>{\centering\arraybackslash}m{.95in}|>{\centering\arraybackslash}m{.95in}|>{\centering\arraybackslash}m{.95in}|>{\centering\arraybackslash}m{1in}}
				
				\hlineB{4}
				&$ G(X) = \{1\}$ & $ G(X)=\{0\} $ & $G(X)=\emptyset$ & $G(X)=\{0,1\}  $\\ \hline
				
				Det. & 
				As the second player, Eloise has a winning strategy & 
				As the second player, Abelard has an winning strategy & 
				As the second player, Eloise and Abelard have a winning strategy & 
				N/A \\ \hline
				
				Non-det. & 
				As the second player, Eloise has a winning strategy& 
				As the second player, Abelard has a winning strategy & 
				As the second player, Eloise and Abelard have a winning strategy & 
				As the second player, Eloise and Abelard have a drawing strategy\\

				\hlineB{4}
			\end{tabularx}
			\vspace{1mm}
			\caption{Winning/drawing strategies for sequential games. N/A denotes a situation that cannot occur.
				We write $G(X) = B$ if $G(X,b) \iff b \in B$.		
		}
		\end{center}
	\end{small}
\label{table:online}
\end{table}

The entries of the table are justified by the following result, along with Theorem~\ref{thm:static-determinacy}:
\begin{thm}
	\label{thm:sequential-determinacy}
	Let $ G  $ be a graph and $ X\subseteq V(G) $ an assignment. We have the following:
	\begin{enumerate}
		\item\label{item:eloise-winning-online} $\neg G(X,0)$ if and only if Eloise has a reactionary winning strategy.
		\item\label{item:abelard-winning-online} $\neg G(X,1)$ if and only if Abelard has a reactionary winning strategy.
%
%
	\end{enumerate}
\end{thm}

\begin{proof}	
	For the left-right implication of \ref{item:eloise-winning-online}, the idea is that Eloise aims to maintain the property that all maximal stable sets consistent with the play thus far intersect the given assignment.
	Formally, given an assignment $X$ intersecting every stable set of a graph $G$, we define what it means to play \emph{according to} $X$, appealing to the classification of maximal stable sets in Propostion~\ref{prop:mc-ms-via-decomposition}:
	\begin{itemize}
		\item If $\tree(G)$ is a singleton there is nothing to play and, since $X$ is nonempty, Eloise wins.
		\item If $\tree(G) = \dis (\tree(G_1), \dots, \tree(G_n)$ then $X$ must intersect some $G_i$, and so Eloise chooses this $G_i$ and continues to play according to $X\cap V(G_i)$.
		\item If $\tree(G) = \con (\tree (G_1), \dots, \tree(G_n))$ then $X$ must intersect every $G_i$, so if Abelard chooses some $G_i$ then Eloise continues to play according to $X \cap V(G_i)$.
		\item If $\tree(G) = (G/P) (\tree(G_1), \dots, \tree(G_n))$ and Abelard plays $T \in \ms(G/P)$, then we must have that $X$ intersects some $M\in T$, and so Eloise chooses any maximal clique of $G/P$ extending $M$.
	\end{itemize}
	Notice that any play of this strategy cannot reach a deadlock, by construction, and so it must terminate in the first case, at a variable, and Eloise wins.
	
	Conversely, let $\sigma$ be a reactionary winning strategy for Eloise and, for any maximal stable set $T$, take the play induced by $\sigma$ when Abelard plays according to $T$ (cf.~the proof of Theorem~\ref{thm:static-determinacy}).
	By induction on the length of the play, it is not difficult to see that the play always intersects $T$, and hence $\neg G(X,0)$, as required.
	
	\ref{item:abelard-winning-online} follows by duality.
\end{proof}

	Notice that the bottom-right entry of the table, when $G$ evaluates to both $0$ and $1$, is inherited directly from the static case, Theorem~\ref{thm:static-determinacy}, and cannot be strengthened for obvious reasons.

		\section{A proof system for entailment via non-linear graphs}
		\label{sect:proof-system}

		In this section we give a complete inference system for Boolean Graph Logic in the style of `deep inference': inference rules may rewrite induced subgraphs of a graph under certain situations.
		In order to admit a complete system, we first (conservatively) extend BGL to account for non-linearity.
		
		\subsection{Nonlinear graphs}
%
		
		We now consider graphs where the same variable may occur many times as a node.
		To avoid ambiguity, the graphs considered until now will now be referred to as `linear' graphs.
We now reserve the set $\Var$ for Boolean variables, equipping graphs with an explicit labelling function assigning a variable to each node:
		
		\begin{definition}
			[(Non-linear) graphs]
			A (labelled) graph is a tuple $G=(V,E,L)$ where $V$ is an arbitrary finite set and, as expected, $E \subseteq \binom V 2$.
			Furthermore $L$ is a function $V \to \Var $.
			For a set $U \subseteq V$, we write $\lfloor U\rfloor := \{ L(v):v\in U \}$.
		\end{definition}
		
		

		We may extend the notions of evaluation and entailment to non-linear graphs in a natural way.
		
		\begin{definition}
Let $G$ and $H$ be (non-linear) graphs and $X \subseteq \Var$.
We define the following notions of evaluation,
			\begin{itemize}
				\item $G(X,1)$ if $\exists S \in \mc(G). \lfloor S \rfloor \subseteq X$.
				\item $G(X,0)$ if $\exists T \in \ms(G). \lfloor T \rfloor \cap X = \emptyset$.
			\end{itemize}
		and the following notions of entailment:
		\begin{itemize}
			\item $G \entails \con H$ if $\forall S \in \mc(G). \exists S'\in \mc (H). \lfloor S'\rfloor \subseteq \lfloor S \rfloor$.
			\item $G \entails \dis H$ if $\forall T \in \ms (H) . \exists T'\in \ms (H) . \lfloor T' \rfloor \subseteq \lfloor T \rfloor$.
		\end{itemize}
		\end{definition}
		
		Clearly these notions admit the anlogous ones for linear graphs as special cases.
		What is more, they are conservative over the analogous notions for non-read-once formulas, when restricted to non-linear cographs.
		We do not go into detail on this point but leave the verification of this fact as an exercise to the reader.
%

		\subsection{Deep inference and rules on modules}
		
		Let us write $G[M]$ to distinguish a module $M$ in a graph $G$.
		We may then write $G[M']$ for the graph where $M$ is replaced by $M'$ in $G$, retaining all the edges connected to the module.
		
		\begin{example}
			Let $G[M]$ be the following graph, with only edges indicated, where $M$ is the module $\{x,y \}$:
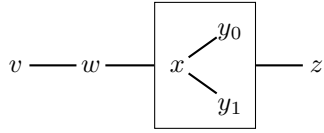
			\[
			\begin{tikzpicture}
			\node[vertex] (v) at (0,0) {$v$};
						\node[vertex] (w) at (1,0) {$w$};

				\node[vertex] (x) at (2.15,0) {$x$};
				\node[vertex] (y0) at (2.85,0.5) {$y_0$};
								\node[vertex] (y1) at (2.85,-0.5) {$y_1$};
				
				\draw[b] (x) -- (y0);
							\draw[b] (x) -- (y1);

				\node[draw,fit=(x) (y0) (y1)] (M) {};
				
			\node[vertex] (z) at (4,0) {$z$};
			
			\draw[b] (v) -- (w);
			\draw[b] (w) -- (M);			
			\draw[b] (M) -- (z);	
					
			\end{tikzpicture}
			\]
			If we set $M'$ to be the graph below, left, then $G[M']$ is the graph below, right:
			\[
			\raisebox{-0.5\height}
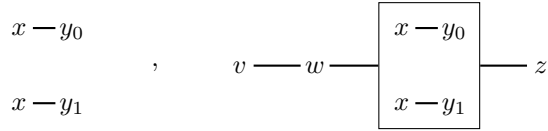
{
		\begin{tikzpicture}
		\node[vertex] (x0) at (2.15,0.5) {$x$};
		\node[vertex] (x1) at (2.15,-0.5) {$x$};
		\node[vertex] (y0) at (2.85,0.5) {$y_0$};
		\node[vertex] (y1) at (2.85,-0.5) {$y_1$};
		
		\draw[b] (x0) -- (y0);
		\draw[b] (x1) -- (y1);
		\end{tikzpicture}	
		}
	\qquad,\qquad
	\raisebox{-0.5\height}{
\begin{tikzpicture}
\node[vertex] (v) at (0,0) {$v$};
\node[vertex] (w) at (1,0) {$w$};

\node[vertex] (x0) at (2.15,0.5) {$x$};
\node[vertex] (x1) at (2.15,-0.5) {$x$};
\node[vertex] (y0) at (2.85,0.5) {$y_0$};
\node[vertex] (y1) at (2.85,-0.5) {$y_1$};

\draw[b] (x0) -- (y0);
\draw[b] (x1) -- (y1);

\node[draw,fit=(x0) (x1) (y0) (y1)] (M) {};

\node[vertex] (z) at (4,0) {$z$};

\draw[b] (v) -- (w);
\draw[b] (w) -- (M);			
\draw[b] (M) -- (z);	

\end{tikzpicture}	
}
			\]
			Formally speaking, the two occurrences of $x$ above are different nodes with the same label.
			When displaying graphs visually, we do not make this distinction explicitly, but formally if we indicate multiple occurrences of the same graph $G$, it means that they are all label-preserving isomorphic to $G$.

			Notice that $M$ and $M'$ are relation webs, and that, say, $M \entails{\con} M'$. We also have that $G[M] \entails{\con} G[M']$.
			This is no coincidence, as we will see in the next result.
		\end{example}
		
		\begin{lemma}
			[Deep inference on graphs]
			\label{lem:deep-inference}
			Suppose $M$ is a module of $G$ and $M \entails \star M'$, for $\star \in \{ \vee, \wedge \}$.
			Then $G[M] \entails \star G[M']$.
		\end{lemma}
	\begin{proof}
		We consider only the case when $\star = \con$, the case of $\star = \dis$ being dual.

Let $S \in \mc(G[M])$. We have the following cases:
\begin{itemize}
	\item If $S \cap M = \emptyset$ then also $S \in \mc (G[M'])$.
	\item Otherwise, we may write $S = S_M \sqcup S'$ for some $S_M \in \mc(M)$, by Lemma~\ref{lem:mc-ms-ser-par-comp}.
	Now, since $M \entails \con M'$, we have some $S_M'\subseteq S_M$ with $S_M'\in \mc(M')$.
	Therefore, again by Lemma~\ref{lem:mc-ms-ser-par-comp}, we have that $S_M'\sqcup S' \in \mc(G[M'])$.
\end{itemize}
Thus indeed $G[M] \entails \con G[M']$.
	\end{proof}
	
The proposition above is	
		a generalisation of `deep inference' reasoning on formulas, where we may operate under arbitrary alternations of $\vee $ and $\wedge$.
%

\begin{definition}
	[Inference rules, systems and derivations]
	An \emph{inference rule} on graphs is simply a binary relation $\rightarrow$ on graphs. 
	A \emph{proof system} is a set of inference rules,\footnote{Note that, in this presentation there is not much difference between a rule and a system, but we maintain the distinction as it is natural from the proof theory and rewriting theory points of view.} and a \emph{derivation} in a proof system is just a sequence of graphs $(G_1, \dots , G_n)$ where each $(G_i,G_{i+1})$ is an instance of a rule in the system.
\end{definition}

Inference rules may be specified in many different ways, and we will introduce some bespoke notation in what follows in order to compactly write inference rules.

		Given the proposition above, we may safely import the standard structural rules from deep inference proof theory, restricted to modules:
		\begin{definition}
			[Structural rules]
The system $\msks$ consists of the following rules:\footnote{The subscripts $l$ and $r$ are usually written as annotations $\uparrow$ and $\downarrow$ respectively, but we chose a different notation to reduce the number of arrows in use. The $l$ and $r$ subscripts is a reference to sides of the sequent calculus.}
	\[
	\begin{array}{rccc}
\wk_r : & 
\Mod{{G_i}}
& \to &
\vbcengraph{\TwoGraphVertMod{G_0,G_1}{g}} \\
\cntr_r : &
\vbcengraph{\TwoGraphVertMod{G,G}{g}} & \to & \Mod G 
	\end{array}
	\quad \quad \quad
	\begin{array}{rccc}
\wk_l : &
\vbcengraph{	\TwoGraphVertMod{G_0,G_1}{r}}
& \to & 
\Mod{G_i}\\
\cntr_l : &
\Mod G &\to & \vbcengraph{\TwoGraphVertMod{G,G}{r}}
	\end{array}
	\]
for $i \in \{0,1\}$.
The fact that the LHS and RHS of these rules are boxed indicates that the corresponding induced subgraphs must be modules in the LHS and RHS, respectively, of an instance of the rule. There may be an ambient surrounding graph that is not indicated, but no nodes outside the module are affected by the rule.
There are no further restrictions on the graphs indicated or the surrounding graph outside the module.
		\end{definition}

For comparison, we give also the formula theoretic versions of some of the rules above:
\begin{equation}
\label{eqn:msks-formula-version}
\begin{array}{rrcl}
\wk_r : & A_i & \to & A_0 \dis A_1 \\
\cntr_r : & A \dis A & \to & A 
\end{array}
\qquad
\begin{array}{rrcl}
\wk_l : & A_0 \con A_1 & \to & A_i \\
\cntr_l : & A& \to & A \con A
\end{array}
\end{equation}
Usually in deep inference proof theory we must operate under some equational theory, here associativity and commutativity of $\dis $ and $\con$.
However this is implicit in the graph theoretic setting.
	
\begin{definition}
	[Soundness and completeness]
We say that an inference rule $\rightarrow$ is \emph{sound} for $\entails \star$, for $\star \in \{\dis,\con \}$, if whenever $G\rightarrow H$ we have $G \entails \star H$.	
A system is sound whenever all its inference rules are.
We say that a system is \emph{complete} for $\entails \star$ if, whenever $G\entails \star H$, there is a derivation from $G$ to $H$.
\end{definition}

Since $\msks$ was induced by a sound system on formulas, we immediately have the following from Lemma~\ref{lem:deep-inference}:
		\begin{proposition}
			\label{prop:msks-sound}
			$\msks$ is sound for both $\entails \con$ and $\entails \dis$.
		\end{proposition}

	\begin{remark}
		[Structural rules beyond modules]
		Despite the fact that we have restricted our deep inference rules to modules, it is not hard to see that there are variations of the rules $\wk_r$ and $\cntr_l$ that operate on arbitrary (induced) subgraphs yet remain sound.
		We do not give details here, being beyond the scope of this work, but we do point out a fundamental problem with extending the other structural rules to non-modular subgraphs, in particular the $\cntr_r$ rule.
		Consider, for example, the following situation:
		\[
		\ThreeGraph{x,y,y}rgg
		\]
		Attempting to apply $\cntr_r$ to the two $y$ nodes would leave us with a choice of how to resolve the clash between the upper $\redge x y$ and the lower $\gedge x y$.
	\end{remark}

	\subsection{Entailment-specific rules}
		Notice that all the rules thus far introduced are sound for both $\entails \dis$ and $\entails \con$. 
	Since we know that these two entailments are distinct, any complete system for either entailment must have additional rules.
	We introduce these in the following definition:
	
	\begin{definition}
		We define the following rules,\footnote{Notice that we could have allowed $x$ to be an arbitrary module $M$ instead, but the following exposition is slightly simpler by restricting to this atomic version.}
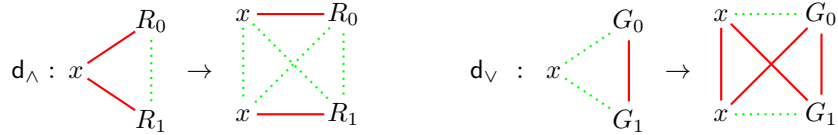
		\[
		\Copy : 
		\ThreeGraph{x,R_0,R_1}rrg
		\ \to \
			\FourGraph{x,R_0,x,R_1}rggggr
		\qquad \qquad \Dopy \ : \ 
		\ThreeGraph{x,G_0,G_1}ggr
		\ \to \ 
		\FourGraph{x,G_0,x,G_1}grrrrg
		\]
		where:
		\begin{itemize}
			\item For $\Copy$, $\{R_0,R_1\}$ partitions the set $\{y :  \{x,y\} \in E(\lhs) \}$, i.e.\ $R_0 \sqcup R_1$ is the set of `red' edges of the LHS including $x$.
			\item For $\Dopy$, $\{G_0,G_1\}$ partitions the set $\{y :  \{x,y\} \notin E(\lhs) \}$, i.e.\ $G_0\sqcup G_1$ is the set of `green' edges of the LHS including $x$.
		\end{itemize}
	Neither the LHS nor the RHS of each rule need form modules, but the indicated edge-relationships must hold. 
	There are no restrictions on the unindicated edges further to the conditions on $R_0, R_1, G_0,G_1$ above.
	On the RHS of both rules there are two occurrences of $x$;
	formally, these are two different nodes with the same label $x$.
	\end{definition}

	We will soon use these rules to achieve normal forms of graphs that drive our ultimate completeness proof.
The point of these rules is not only that they are sound for $\entails \con$ and $\entails \dis$, respectively, but moreover that they do not change the Boolean relation computed, as shown in the following result.

			\begin{lemma}
				\label{lem:Copy-Dopy-sound}
			We have the following:
			\begin{enumerate}
				\item\label{item:copy-pres-mcs} If $G \annarr{\Copy} H$ then $\lfloor \mc (G) \rfloor  = \lfloor \mc (H) \rfloor$, i.e.\ $G $ and $H$ have the same maximal cliques, up to the variables occurring in them.
				\item\label{item:dopy-pres-mss} If $G \annarr{\Dopy} H$ then $\lfloor \ms (G) \rfloor  = \lfloor \ms (H) \rfloor$, i.e.\ $G $ and $H$ have the same maximal stable sets, up to the variables occurring in them.
			\end{enumerate}
		\end{lemma}

	\begin{proof}
		We prove only \ref{item:copy-pres-mcs}, the case of \ref{item:dopy-pres-mss} being dual.
		
Consider the corresponding instance of $\Copy$,
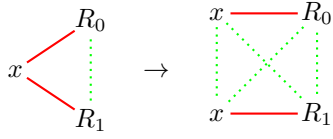
\[
\ThreeGraph{x,R_0,R_1}rrg
\quad \to \quad
\FourGraph{x,R_0,x,R_1}rggggr
\]
and let us call the node labelled $x$ on the LHS $v$, the upper node labelled $x$ on the RHS $v_0$, and the lower node labelled $x$ on the RHS $v_1$. Assume all other nodes of the RHS have the same name as their corresponding nodes on the LHS.

Suppose $S \in \mc(\lhs)$, and notice that $S$ must be disjoint from either $R_0$ or $R_1$, since $\gedge{R_0}{R_1}$.
We define $S' \in \mc(\rhs)$ with $\lfloor S' \rfloor = \lfloor S \rfloor$ as follows:
\begin{itemize}
	\item If $v \notin S$ then we define $S' = S$;
		\item If $v \in S$ and $S\cap V(R_0) = \emptyset$, then we define $S' = (S \setminus\{v\}) \cup \{v_1\}$;
		\item If $v \in S$ and $S\cap V(R_1) = \emptyset$, then we define $S' = (S \setminus\{v\}) \cup \{v_0\}$.
\end{itemize}
By construction we have that $S \cong S'$, so $S'$ is a clique, and maximality is immediate.

For the converse direction, notice that any maximal clique of the RHS also cannot intersect both $R_0$ and $R_1$, and furthermore cannot contain both $v_0$ and $v_1$, since $\gedge {v_0} {v_1}$.
Thus the mapping from $S$ to $S'$ above is a bijection $\mc(\lhs) \to \mc(\rhs)$, finishing the proof.
%
		\end{proof}

	The following is an immediate consequence of the preceding lemma:
	\begin{proposition}
		\label{prop:copy-dopy-sound}
		We have the following:
		\begin{enumerate}
			\item 	Both $\Copy $ and $\Copy^{-1}$ are sound for $\entails \con$.
			\item 	Both $\Dopy $ and $\Dopy^{-1}$ are sound for $\entails \dis$.
		\end{enumerate}
		
	\end{proposition}

\subsection{Reductions to DNF and CNF}

	Our completeness strategy is motivated by the disjunctive and conjunctive normal forms of Boolean functions.
In fact, Boolean relations may naturally be associated with a DNF, determining evaluation to $1$, and a CNF, determining evaluation to $0$.
Let us frame these normal forms in a graph-theoretic context.		
\begin{definition}
	[DNFs and CNFs]
	A DNF is a graph where all maximal cliques are disjoint. A CNF is a graph where all maximal stable sets are disjoint.
\end{definition}
Of course, DNFs are just (non-linear) relation webs of formulas of the form $\bigvee_i \bigwedge_j x_{ij}$ and CNFs are just (non-linear) relation webs of formulas of the form $\bigwedge_i \bigvee_j x_{ij}$.
We will thus identify DNFs and CNFs with their formula representations when it is convenient.

	The main result of this subsection is that we may generate DNFs and CNFs of arbitrary graphs using our entailment-specific rules:
\begin{lemma}
	\label{lem:dnf-cnf}
	For any graph $G$ we have the following:
	\begin{enumerate}
		\item There is a DNF $A$ with $G\annarr{\Copy}^* A$
		\item There is a CNF $B$ with $G \annarr{ \Dopy}^* B$.
	\end{enumerate}
\end{lemma}
In what follows, we will only concern ourselves with DNFs and completeness for $\entails \con$, with the case of CNFs and $\entails \dis$ following by duality.

Notice that the rule $\Copy$ bears semblance to the distibutivity rule on formulas:
\[
\Copy' : \quad A\con (B\dis C) \quad \entails{} \quad (A\con B) \dis (A \con C)
\]
Indeed, the rule above is a special case of $\Copy$ when all nodes are indicated: if $\Copy' : A \to B$ then $\Copy : \web(A) \to \web(B) $.\footnote{This holds also for deep versions of $\Copy'$ if we allow $\Copy$ to operate within modules too.}
On the other hand, $\Copy$ is strictly more general than $\Copy'$, as shown in the following example.
\begin{example}
	\label{ex:copy-cograph-to-p5}
	Consider the following instance of $\Copy$:
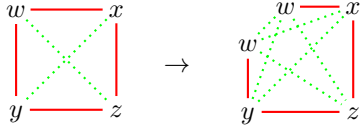
	\[
	\FourGraph{w,x,y,z}rrggrr
	\quad \to \quad
	\raisebox{-0.4\height}{
\begin{tikzpicture}
\node[vertex] (w1) at (0.5,0) {$w$};
\node[vertex] (x) at (1.4,0) {$x$};
\node[vertex] (w2) at (0, -0.5) {$w$};
\node[vertex] (y) at (0,-1.4) {$y$};
\node[vertex] (z) at (1.4,-1.4) {$z$};

\draw[r] (w1) -- (x);
\draw[r] (w2) -- (y);
\draw[r] (x) -- (z);
\draw[r] (y) -- (z);
\draw[g] (w1) -- (w2);
\draw[g] (w1) -- (y);
\draw[g] (w1) -- (z);

\draw[g] (w2) -- (z);
\draw[g] (w2) -- (x);

\draw[g] (x) -- (y);
\end{tikzpicture}	
}
	\]
	Notice that the LHS is the relation web of the formula $(w\dis z) \con (x \dis y)$, while the RHS is isomorphic to the $P_5$, and so corresponds to no formula.
\end{example}

However, it is not hard to see that DNFs remain the only normal forms of $\Copy$:
\begin{observation}
	\label{obs:nf-dnf-cnf}
	The only normal forms of $\Copy$ are DNFs, and the only normal forms of $\Dopy$ are CNFs.
\end{observation}
\begin{proof}[Proof sketch]
	Any graph that is not a DNF has two intersecting maximal cliques, which would form a redex for $\Copy$.
	The argument for $\Dopy$ is dual.
\end{proof}

\begin{example}
	Revisiting the above Example~\ref{ex:copy-cograph-to-p5}, we may continue applying $\Copy$ as follows (now with only edges indicated):
	\[
	\raisebox{-0.4\height}
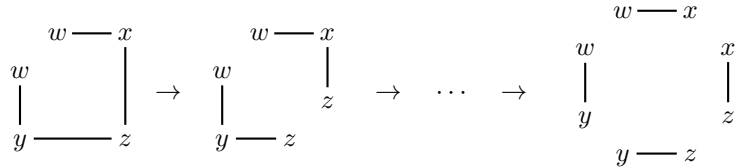
{
		\begin{tikzpicture}
		\node[vertex] (w1) at (0.5,0) {$w$};
		\node[vertex] (x) at (1.4,0) {$x$};
		\node[vertex] (w2) at (0, -0.5) {$w$};
		\node[vertex] (y) at (0,-1.4) {$y$};
		\node[vertex] (z) at (1.4,-1.4) {$z$};
		
		\draw[b] (w1) -- (x);
		\draw[b] (w2) -- (y);
		\draw[b] (x) -- (z);
		\draw[b] (y) -- (z);
		\end{tikzpicture}}	
		\ \to \ 
		\raisebox{-0.4\height}{
			\begin{tikzpicture}
			\node[vertex] (w1) at (0.5,0) {$w$};
			\node[vertex] (x) at (1.4,0) {$x$};
			\node[vertex] (w2) at (0, -0.5) {$w$};
			\node[vertex] (y) at (0,-1.4) {$y$};
			\node[vertex] (z1) at (0.9,-1.4) {$z$};
			\node[vertex] (z2) at (1.4,-0.9) {$z$};
			
			\draw[b] (w1) -- (x);
			\draw[b] (w2) -- (y);
			\draw[b] (x) -- (z2);
			\draw[b] (y) -- (z1);
			\end{tikzpicture}}	
		\quad \to \quad \cdots \quad \to \quad
		\raisebox{-0.4\height}{
			\begin{tikzpicture}
			\node[vertex] (w1) at (0.5,0) {$w$};
			\node[vertex] (x1) at (1.4,0) {$x$};
					\node[vertex] (x2) at (1.9,-0.5) {$x$};
			\node[vertex] (w2) at (0, -0.5) {$w$};
			\node[vertex] (y1) at (0,-1.4) {$y$};
					\node[vertex] (y2) at (0.5,-1.9) {$y$};
			\node[vertex] (z1) at (1.4,-1.9) {$z$};
			\node[vertex] (z2) at (1.9,-1.4) {$z$};
			
			\draw[b] (w1) -- (x1);
			\draw[b] (w2) -- (y1);
			\draw[b] (x2) -- (z2);
			\draw[b] (y2) -- (z1);
			\end{tikzpicture}}
	\]
	The first step applies $\Copy $ on $z$, and then there are two steps applying $\Copy$ to $x$ and $y$ in any order.
	The resulting graph is a DNF.
	There is also a derivation from the original formula $(w\dis z) \con (x \dis y)$ to DNF by continuously applying $\Copy'$ (using deep inference), but by definition that derivation only includes cographs, unlike the one above.
\end{example}

We are now ready to prove the main result of this subsection.
\begin{proof}
	[Proof of Lemma~\ref{lem:dnf-cnf}]
	We prove only (1), since (2) follows by duality.
	By Observation~\ref{obs:nf-dnf-cnf} above, it suffices to show that the rule $\Copy $ is \emph{terminating}, i.e.\ that any $\Copy$-derivation has finite length.
	

Recall from the proof of Lemma~\ref{lem:Copy-Dopy-sound} that the number of maximal cliques remains constant for any instance of $\Copy$. 
However, we have strictly reduced the number of intersections, since any maximal cliques intersecting at the duplicated variable become disjoint on the RHS.
Thus any $\Copy$ derivation must terminate, as required.
\end{proof}

		\subsection{A completeness result}
We are now ready to give our systems for both entailment relations.

\begin{definition}
	[Systems for entailment]
	We define the following systems of inference rules:
	\begin{itemize}
		\item $ \consys \ := \ \msks \cup \{ \Copy, \Copy^{-1} \}$.
			\item $\dissys \ := \ \msks \cup \{ \Dopy, \Dopy^{-1} \}$.
	\end{itemize}
\end{definition}

Notice that it is immediate from Propositions~\ref{prop:msks-sound} and \ref{prop:copy-dopy-sound} that $\consys$ and $\dissys$ are sound for $\entails \con$ and $\entails \dis$ respectively.
The main goal of this section is to establish the converse result:

		\begin{theorem}
			[Completeness]
			\label{thm:completeness}
		We have:
		\begin{enumerate}
		\item If $G \entails \wedge H$ then $G \consys H$.
		\item if $G \entails \vee H$ then $G\dissys H$.
		\end{enumerate}
		\end{theorem}

	We are almost ready to prove this, but we need the 
	following intermediate result, which is well-known in deep inference proof theory:
		\begin{lemma}
			[Completness for DNFs and CNFs]
			\label{lem:completeness-dnf-cnf}
			We have the following:
			\begin{enumerate}
				\item If a DNF $A$ logically implies a DNF $B$ then there is a $\msks $ derivation $A \rightarrow^* B$.
				\item If a CNF $A$ logically implies a CNF $B$ then there is a $\msks $ derivation $A \rightarrow^* B$.
			\end{enumerate}
		\end{lemma}
	\begin{proof}
		We only prove (1), the case of (2) being dual.
		Since the statement only concerns formulas and any rule instance of $\msks$ preserves cograph-ness, we will present the argument in terms of formulas, cf.~\eqref{eqn:msks-formula-version}.
		
		Suppose we have a valid implication,
\begin{equation}
\label{eqn:dnf-implies-dnf}
\bigvee\limits_{i\in I} \bigwedge S_i \entails{} \bigvee\limits_{j\in J} \bigwedge S_j'
\end{equation}
		 for some sets of variables $S_i$ and $S_j'$.
		 Recall that we each $\bigwedge S_i$ and $\bigwedge S_j'$ are called \emph{terms}.
		Appealing to the definition of $\entails \con$, there must be some function $f:I \to J$ such that $\lfloor S_{f(i)}' \rfloor \subseteq \lfloor S_i \rfloor$.
We thus have derivations:
\begin{equation}\label{eqn:dnf-completeness-coweakening}
\{\wk_l, \cntr_l \} \ : \ \bigwedge S_i \ \to^* \ \bigwedge S_{f(i)}'
\end{equation}
($\cntr_l$ is required in case there are multiple occurrences of variables in the RHS).
Applying these to each term of LHS\eqref{eqn:dnf-implies-dnf} yields a DNF $\bigvee\limits_{i \in I} \bigwedge S_{f(i)}'$.
To arrive at RHS\eqref{eqn:dnf-implies-dnf}, we need to use the dual rules:
\begin{itemize}
	\item If $f $ is not surjective, e.g.\ $\forall i \in I . f(i) \neq j$, then we may apply $\wk_r$ to recover the missing terms, e.g.:
	\[
	\wk_r \ : \ \bigvee\limits_{i \in I} \bigwedge S_{f(i)}' \ \to \ \bigvee\limits_{i \in I} \bigwedge S_{f(i)}' \ \dis \ \bigwedge S_j' 
	\]
	\item If $f$ is not injective, e.g.\ $f(i) = f(i' ) =j$ for some distinct $i,i'$, then we may apply $\cntr_r $ to remove duplicate terms, e.g.:
	\[
	\cntr_r \ : \ \bigwedge S_j \dis \bigwedge S_{j}' \ \to \ \bigwedge S_j
	\]
\end{itemize}
We have thus derived $\eqref{eqn:dnf-implies-dnf}$ in $\{\wk_l, \cntr_l, \wk_r, \cntr_r \}$, as required.
		\end{proof}
		
We are now ready to prove the main completeness result:
		\begin{proof}
			[Proof of Theorem~\ref{thm:completeness}]
We prove only (1), the case of (2) following by duality.
			Suppose $G\entails \con H$, and by Lemma~\ref{lem:dnf-cnf} let $A$ and $B$ be DNFs such that:
			\begin{equation}
			\label{eqn:completeness-dnf-LHS}
			G \annarr{ \Copy}^* A
			\end{equation}
			\begin{equation}
			\label{eqn:completeness-dnf-RHS}
			H \annarr{\Copy}^* B
			\end{equation}
			Notice that, since $G$ and $A$ have the same maximal cliques up to variables occurring, by Lemma~\ref{lem:Copy-Dopy-sound}, we have that $G$ and $A$ are equivalent. Similarly for $H$ and $B$ so, since $G\entails \con H$, we also have that $A $ logically implies $B$.
			We may thus build the following $\consys$-derivation:
			\[
			\begin{array}{rcll}
			G & \annarr{\Copy}^* & A & \text{by \eqref{eqn:completeness-dnf-LHS}} \\
			 & \annarr{\msks}^* & B & \text{by Lemma~\ref{lem:completeness-dnf-cnf}} \\
			  & \annarr{\Copy^{-1}}^* & H & \text{by \eqref{eqn:completeness-dnf-RHS}} 
			\end{array}
			\]
			This concludes the proof.
%
%
		\end{proof}

	\todo{example with interpolant?}
	
	\subsection{On atomicity and linearity}
	\label{subsect:atomicity-linearity}
Though the structural proof theory of Boolean Graph Logic is beyond the scope of this paper, we make some observations here that suggest that BGL should enjoy decomposition theorems similar to those of deep inference proof theory \cite{BrunTiu:01:A-Local-:mz}.
We do not give a formal development, saving that for future work.

		A key feature of deep inference is that the contraction rules, $\cntr_l$ and $\cntr_r $ may be reduced to atomic form.
		This is thanks to the \emph{medial} rule:
		\[
		\medial \ :\  (A \con B) \dis (C \con D) \ \to \ (A\dis C) \con (B\dis D)
		\]
		For instance, we may transform a $\cntr_r$ inference $(A\con B) \dis (A\con B) \to (A \con B)$ into one with smaller contraction redexes as follows:
		\[
		\begin{array}{rcl}
		\underline{(A \con B) \dis (A\con B)} & \annarr{\medial} & \underline{(A \dis A)} \con (B\dis B) \\
		& \annarr{\cntr_r} & A \con \underline{(B\dis B)} \\
		& \annarr{\cntr_r} & A \con B
		\end{array}
		\]
		The reduction for $\cntr_l$ is dual.
		A similar reduction can be carried out in the graph theoretic setting by introducing the following medial rule for prime graphs $P[v_1, \dots , v_n]$ (all nodes indicated):
		\[
		P[G_1, \dots, G_n] \dis P[G_1', \dots, G_n'] \quad \to \quad P[\gedge{G_1}{G_1'}, \dots, \gedge{G_n}{G_n'}]
		\]
		It is not hard to verify that this rule is sound for both $\entails \con $ and $\entails \dis$, and also allows us to reduce contraction inferences to atomic form, though we omit the details here.

The medial rule above is an example of a \emph{linear} rule: it does not duplicate or erase any nodes in a graph. 
At the level of formulas such rules are important, since decomposition theorems of deep inference are typically agnostic about the choice of linear rules, once all the structural rules have been made atomic, cf.~\cite{GuglGund:07:Normalis:lr}.
One issue for our system $\consys$ (and $\dissys$) is that the rule $\Copy$ (dually $\Dopy$) is not linear.
However, it turns out that we may decompose it into atomic contraction inferences and the following linear rule:\footnote{This observation was made by Lutz Strassburger.}
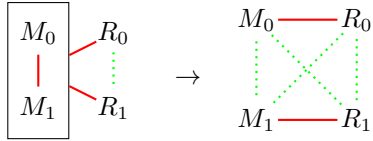
\[
\raisebox{-0.4\height}{
\begin{tikzpicture}
\node[vertex] (G0) at (0,0) {$M_0$};
\node[vertex] (G1) at (0,-1) {$M_1$};
\draw[r] (G0) -- (G1);

\node[draw,fit=(G0) (G1) ] (M) {};

\node[vertex] (R0) at (1,0) {$R_0$};
\node[vertex] (R1) at (1,-1) {$R_1$};

\draw[r] (M) -- (R0);
\draw[r] (M) -- (R1);
\draw[g] (R0) -- (R1);
\end{tikzpicture}
}
\quad \to \quad
		\FourGraph{M_0, R_0, M_1, R_1}rggggr
\]
where $\{R_0,R_1\}$ partition the set of edges to the indicated module on the LHS.

It would be interesting to establish whether the incorporation of these linear rules (and perhaps others) and atomisation of the structural rules could lead to a well-behaved proof theory on arbitrary graphs, similar to deep inference.

		\section{Conclusions}
		
		In this work we presented a graph theoretic extension of Boolean logic that we called Boolean Graph Logic (BGL).
		BGL extended the semantics of Boolean logic from Boolean functions to more general Boolean relations, and we recovered a decomposition of entailment into two dual notions.
		The purpose of this article was to establish the fundamental theory behind BGL, for which we gave perspectives via complexity (Section~\ref{sect:complexity}), games (Section~\ref{sect:games}) and proofs (Section~\ref{sect:proof-system}).
		
		In future work we are most interested in developing the structural proof theory of BGL, building on the discussion of Section~\ref{subsect:atomicity-linearity} to establish \emph{decomposition theorems}, \`a la deep inference. 
		Such an investigation would also help compare BGL with the approach of \cite{AccHorStr:20}, complementary to a semantic investigation. Conversely, it would be interesting to see how the logic of \cite{AccHorStr:20} may be extended by structural rules, which can be problematic for the `splitting' technique there used.
		
		It would also be interesting to examine how to incorporate forms of negation and implication natively into BGL.
		For this it would be natural to consider the behaviour of analogous connectives from Computability Logic, cf.~\cite{Japaridze17}.

		\bibliographystyle{alpha}
		\bibliography{graph-refs,timsbib}
		
		\appendix

		\section{Deterministic and total graphs are $P_4$-free}
		\label{app:det-tot-pfourfree}
		In this section we give a self-contained proof of Theorem~\ref{thm:function-p4free}, showing that a graph that is deterministic and total is $ P_4 $-free. We need to introduce an additional concept first.
		
		\begin{defi}
			Let $ G $ be a graph and $ Y \subseteq V(G) $. A \e{selection} with respect to $ Y $ is a set $ Sel = \{T_x \mid T_x \in MS(G) \text{ and } T_x \cap Y = \{x\}	\} $.\\
			We call a selection w.r.t $ Y $ a \e{covering} if there is a $ D \in MS(G) $ with $ D \subseteq \cup_{x \in Y}T_x $ and $ D \cap Y = \emptyset $.\\
			We call a selection w.r.t $ Y $ a \e{non-covering} if it is not a covering.
		\end{defi}
		
		\begin{ex}
			If we choose $ Y $ so that $ Y $ is a clique, then a covering always exists. On the other hand, if we  choose $ Y = V(G) $, then a selection w.r.t $ Y $ does not exist. Take the following simple example:
			Let $ G $ be the following graph: \[ \FourGraph{a,b,c,d}rrrggr \]
			Then $ G $ is deterministic, and thus $ CIS $. Let $ Y = \set{a,c,d} $. Then $ Y $ is a maximal clique, and a selection exists with $ Sel = \set{ \set{a}, \set{b,c}, \set{b,d}}  $.\\
			
			For a less trivial example, take the following graph:
			\[ \SixGraph{a,b,d,c,e,f}rgrgrgrgrrrgggg \]
			
			Let $ Y = \set{c,e} $. Then the only stable set intersecting $ c $ is $ \set{c,e,f} $, but that also intersects $ e $, so there is no selection. 	
			Now let $ Y = \set{b,c,d} $. Then we have the selection $ Sel = \set{ \set{b,d}, \set{c,e,f} , \set{d,a}} $. Not only is it a selection, it is also covering, because we have $ \set{a,e} \subseteq V(G) = \set{b,d} \cup \set{c,e,f} \cup \set{d,a} $, and $ \set{a,e} \in MS(G) $.
		\end{ex}
		
		The following Lemma sheds some more light on coverings, and is the key piece to proving the theorem:
		
		\begin{lem}
			Let $ G $ be a total graph, and $ Y \subseteq V(G) $ such that $ Y $ is not a clique. Then every selection w.r.t $ Y $ is covering.
		\end{lem}
		\begin{proof}
			We prove the contrapositive. Assume there is a non-covering selection $ Sel = \{T_x \mid T_x \in MS(G) \text{ and } T_x \cap Y = \{x\} \}$ w.r.t. $ Y $. Define the set 
			\[ B:= (V(G) \setminus \bigcup_{x \in Y} T_x ) \cup Y \]
			Notice that $ B $ intersects every $ T \in MS(G) $: If $ T $ is a maximal stable set that intersects $ Y $, then, as $ Y \subseteq B$, $ T $ also intersects $ B$. If $ T $ is a maximal stable set that doesn't intersect $ Y $, then, because $ Sel $ is not a covering, we have $ T \nsubseteq \cup_{x \in Y}T_x $, so $ T$ intersects $ B $, by the definition of $ B $. \\
			There is no maximal stable set disjoint from $ B $, so $ G $ doesn't evaluate $ B $ to $ 0 $. $ G $ is total, so $ e_G(B,1) $, i.e. there exists a $ S \in MC(G)$ with $ S \subseteq B $.\\
			
			By the definition of the $ T_x$'s, we get $ T_x \cap Y = \{x\} $, and by the definition of $ B $ therefore also $ T_x \cap B = \{x\}$ for every $ x \in Y $. We have $ T_x \in MS(G)$ for every $ x \in Y $, and $ S \in MC(G) $, so, because $ G $ is deterministic, by the CIS property we get $ |S \cap T_x| = 1 $ for all $ x \in Y $. Because $ S $ is contained in $ B $, we get $ S \cap T_x = \{x\} $ for all $ x \in Y $. So $ Y \subseteq S $.
		\end{proof}
		
		\begin{ex}
			Take the following graph, and let $ Y = \set{c,f} $:
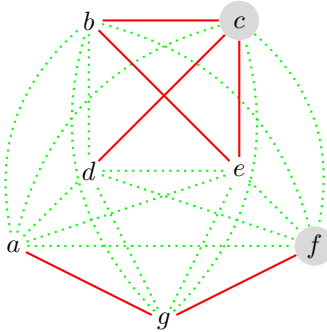
			\[ \SevenGraphModule{a,b,c,d,e,f,g}gggggrgrgrrggggrggggr \]
			
			Then $ Y $ is not a clique. 
			We have $ MS(G) = \set{  \set{g,c}, \set{gde},\set{gdb},\set{acf},\set{abdf},\set{aedf}}    $. Then with $ T_c = \set{g,c}$ , $T_f = \set{a,b,d,f}$, the set  $Sel = \set{ T_c , T_f } $ is a selection. 
			
			We have $ \cup_{x \in Y}T_x = T_c \cup T_f = \set{a,b,c,d,f} $.
			Also, $ \set{g,d,b} $ is a maximal stable set with $\subseteq \set{a,b,c,d,f}$ and $ \set{g,d,b} \cap Y = \emptyset $. Therefore, $ Sel $ is not only a selection, but also a covering.

		\end{ex}

		We now have all the results we need to prove the characterisation of Boolean functions.

		\begin{proof}
			[of Theorem~\ref{thm:function-p4free}]
			We prove the left-right implication.
			Assume $ G $ has a $ P_4 $ generated by the nodes $ \{a,b,c,d\} $ like seen below.
			\[ \FourGraph{b,c,a,d}rrggrg \]
			We can extend the edges $ \{a,b\}, \{c,d\} $ to maximal cliques $ S_1,S_2 \in MC(G)$. 
			
			We write $ S_1 = \{a,b\} \sqcup S_{ab} $, $ S_2 = \{c,d\} \sqcup S_{cd} $, where $ S_{ab} , S_{c,d} $ are (possibly empty) sets of nodes, and $ \sqcup $ denotes the disjoint union.
			Likewise, we extend the edges $ \{a,c\}, \{b,d\}  $ to maximal stable sets $ T_1, T_2 $, and write $ T_1 = \{a,c\} \sqcup T_{ac} $, $ T_2 = \{b,c\} \sqcup T_{bd} $. \\
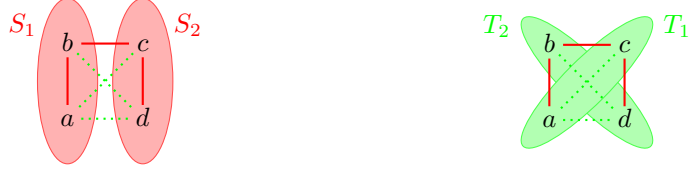
			\begin{minipage}{0.5\textwidth}
				\[	\begin{tikzpicture}
				\draw[red!70,fill=red!30] (0,0.5) ellipse (0.4cm and 1.1cm);
				\draw[red!70,fill=red!30] (1,0.5) ellipse (0.4cm and 1.1cm);
				\node[vertex]  (v1) at (0, 1) {$b$};
				\node[vertex]  (v2) at (1, 1) {$c$};
				\node[vertex] (v3) at (0, 0) {$a$};
				\node[vertex]  (v4) at (1, 0) {$d$};
				\draw[r]  (v1) -- (v2) ;
				\draw[r]  (v1) -- (v3);
				\draw[g] (v1) -- (v4);
				\draw[g]  (v2) -- (v3);
				\draw[r] (v2) -- (v4);
				\draw[g]  (v3) -- (v4);
				\node[text=red] (S1) at (-0.6,1.25) {$S_1$};
				\node[text=red] (S2) at (1.6,1.25) {$S_2$};
				
				\end{tikzpicture}  \]
			\end{minipage}
			\begin{minipage}{0.5\textwidth}
				\[
				\begin{tikzpicture}
				\draw[green!70,fill=green!30,rotate around={45:(0.5,0.5)},] (0.5,0.5) ellipse (0.3cm and 1.2cm);
				\draw[green!70,fill=green!30,rotate around={-45:(0.5,0.5)},] (0.5,0.5) ellipse (0.3cm and 1.2cm);
				\node[vertex]  (v1) at (0, 1) {$b$};
				\node[vertex]  (v2) at (1, 1) {$c$};
				\node[vertex] (v3) at (0, 0) {$a$};
				\node[vertex]  (v4) at (1, 0) {$d$};
				\draw[r]  (v1) -- (v2) ;
				\draw[r]  (v1) -- (v3);
				\draw[g] (v1) -- (v4);
				\draw[g]  (v2) -- (v3);
				\draw[r] (v2) -- (v4);
				\draw[g]  (v3) -- (v4);
				\node[text=green] (T1) at (1.7,1.25) {$T_1$};
				\node[text=green] (T2) at (-0.7,1.25) {$T_2$};
				\end{tikzpicture}  \]
			\end{minipage}
			\vspace{1mm}

			Notice that by the CIS property, we get $ |S_1 \cap T_1| = |S_1 \cap T_2| = |S_2 \cap T_1| = |S_2 \cap T_2| = 1$, so we have $ S_1 \cap T_1 = \{a\} $, $ S_1 \cap T_2 = \{b\} $, $ S_2 \cap T_1 = \{c\} $, $ S_2 \cap T_2 = \{d\} $.
			Therefore, we can easily check that the three sets $ i) \hspace{1mm} \{a,b,c,d\} $, $ ii) \hspace{1mm} S_{ab} \cup S_{cd} $, $ iii) \hspace{1mm} T_{ac}\cup T_{bd} $ are all pairwise disjoint:\\
			
			We show that $ i), ii) $ are disjoint:\\	
			By definition, $ a, b \notin S_{ab} $. Without loss of generality, assume $ a \in S_{cd}$. Then $ |S_2 \cap T_1|=2 $, which is a contradiction. So $ a,b \notin S_{cd}$. A completely symmetric argument shows that $ c,d \notin S_{ab} $, and therefore the two sets are disjoint.\\
			The sets $ i),iii) $ are disjoint by the exact same argument.\\
			
			To show that $ ii),iii) $ are disjoint, notice that 
			\[ (S_{ab} \cup S_{cd}) \cap (T_{ac} \cup T_{bd}) = 
			(S_{ab} \cap T_{ac}) \cup 
			(S_{ab} \cap T_{bd}) \cup 
			(S_{cd} \cap T_{ac}) \cup 
			(S_{cd} \cap T_{bd})   \]
			Due to CIS, the only candidates for these intersections would be $ a,b,c,d $, but by the previous observations, these cannot be contained in the intersection. Thus, each of the four intersections must be empty, and therefore $ ii),iii) $ must be disjoint.\\

			The set $ \{a,d\} $ is not a clique, so by the previous lemma, every selection with respect to it is covering. Notice that the set  $ Sel = \{T_1, T_2\}  $ is a selection w.r.t to $ \{a,d\} $ and is therefore covering. 
			So there is a $ D \in MS(G)  $ with $ D \cap \{a,d\} = \emptyset $ and $ D \subseteq T_1 \cup T_2 $. By the previous observation, we have $ (T_1 \cup T_2)\cap S_1 = \{a,c\}$, and $(T_1 \cup T_2) \cap S_2 = \{b,d\} $. 
			$ D $ is a maximal stable set, so by CIS, $ |D \cap S_1| = |D \cap S_2| = 1 $, so, because $ a,d \notin D $, we get  $ D \cap S_1 = \{c\} $, and $ D \cap S_2 =\{b\} $.

			So we get $ b,c \in D \in MS(G) $, which is a contradiction, because there is a red edge between $ b $ and $ c $. 
		\end{proof}

	\end{document}